\newif\ifC

\ifC
    \documentclass[a4paper,UKenglish,cleveref, autoref, thm-restate]{lipics-v2021}
\else
    \documentclass[11pt]{article}
\fi

\ifC
\nolinenumbers
\else
\fi

\ifC
\author{Avi Kadria}
{Department of Computer Science, Bar-Ilan University, Ramat Gan 5290002, Israel}
{avi.kadria3@gmail.com}
{https://orcid.org/0000-0001-8449-3284}
{}

\author{Liam Roditty}
{Department of Computer Science, Bar-Ilan University, Ramat Gan 5290002, Israel}
{liam.roditty@biu.ac.il}
{https://orcid.org/0000-0002-5289-198X}
{Supported in part by BSF grants 2016365 and 2020356.}

\authorrunning{A. Kadria, L. Roditty}
\Copyright{Avi Kadria, Liam Roditty}



\category{Track A: Algorithms, Complexity and Games}

\relatedversion{}
\relatedversiondetails{Full Version}{https://arxiv.org/abs/2507.06721}

\EventEditors{Sayan Bhattacharya, Danupon Nanongkai, Michael Benedikt, and Gabriele Puppis}
\EventNoEds{4}
\EventLongTitle{53rd International Colloquium on Automata, Languages, and Programming (ICALP 2026)}
\EventShortTitle{ICALP 2026}
\EventAcronym{ICALP}
\EventYear{2026}
\EventDate{July 7--10, 2026}
\EventLocation{Royal Holloway, University of London, Egham, United Kingdom}
\EventLogo{}
\SeriesVolume{374}
\ArticleNo{105}

\else
\author{Avi Kadria\thanks{Department of Computer Science, Bar-Ilan University, Ramat Gan 5290002, Israel. E-mail {\tt avi.kadria3@gmail.com}.} \and Liam Roditty\thanks{Department of Computer Science, Bar-Ilan University, Ramat Gan 5290002, Israel. E-mail {\tt liam.roditty@biu.ac.il}. Supported in part by BSF grants 2016365 and 2020356.}}
\fi

\ifC
\else
\usepackage[margin=1in]{geometry}
\usepackage{setspace}
\fi

\usepackage{multirow}
\usepackage{makecell}
\ifC \else
\usepackage[bottom]{footmisc}
\fi

\usepackage{epigraph}
\usepackage{graphicx}
\usepackage{tikz}
\usetikzlibrary{arrows,automata,patterns,decorations.pathreplacing,decorations.pathmorphing,calc}
\usepackage{enumitem}
\usepackage{hyperref}
\usepackage{lipsum}

\usepackage[linesnumbered,ruled,vlined,boxed,algo2e]{algorithm2e}
\usepackage{nicefrac}
\usepackage{amsmath}
\usepackage{float}
\usepackage{mathtools}
\usepackage{amsthm}
\ifC\else
\usepackage{fullpage}
\fi
\usepackage{amsfonts}
\usepackage{clipboard}

\ifC
\else
\usepackage[nameinlink]{cleveref}
\newtheorem{theorem}{Theorem}

\newtheorem{lemma}{Lemma}

\fi
\newtheorem{problem}{Problem}

\newtheorem{cclaim}{Claim}[lemma]

\newtheorem{property}[lemma]{Property}

\Crefname{enumi}{(item)}{(items)}

\newcommand{\Reminder}[1]{

\vspace{0.5em}
\noindent\textbf{Reminder of~\autoref{#1}.} \textit{\Paste{#1}}
\vspace{0.5em}

}

\usepackage{xcolor}
\definecolor{DarkGreen}{RGB}{1,50,32}
\usepackage{pgfplots}
\pgfplotsset{compat=1.18}

\usepackage{silence}
\usepackage{chngcntr}
\counterwithin*{equation}{section}
\counterwithin*{equation}{subsection}

\begin{document}
\ActivateWarningFilters[pdftoc]
\newcommand{\defeq}{:=}
\newcommand{\eps}{\varepsilon}

\newcommand{\liam}[1]{{\color{red} \textbf{Liam}: #1}} 
\newcommand{\avi}[1]{{\color{purple} \textbf{Avi}: #1}} 
\newcommand{\new}[1]{{\color{blue} #1}}

\newcommand{\blue}[1]{{\color{blue}#1}}

\newcommand{\ReturnCode}{\textbf{return}}
\newcommand{\codestyle}[1]{\texttt{#1}}
\newcommand{\Initialize}{\mbox{\codestyle{Initialize}}}
\newcommand{\Cycle}{\mbox{\codestyle{Cycle}}}
\newcommand{\ADO}{\mbox{\codestyle{ADO}}}
\newcommand{\hADO}{\mbox{\codestyle{hADO}}}
\newcommand{\Query}{\mbox{\codestyle{.Query}}}
\newcommand{\ADOQuery}{\mbox{\codestyle{ADO.Query}}}
\newcommand{\ConstructADO}{\codestyle{ConstructADO}}
\newcommand{\FastAPSP}{\codestyle{FastAPSP}}
\newcommand{\Construct}{\mbox{\codestyle{Construct}}}
\newcommand{\Spanner}{\codestyle{Spanner}}
\newcommand{\RT}{\mbox{\codestyle{RT}}}
\renewcommand{\L}{\mbox{\codestyle{L}}}
\newcommand{\CycleOdd}{\codestyle{CycleOdd}}
\newcommand{\BallOrCycle}{\codestyle{BallOrCycle}}
\newcommand{\ClusterOrCycleBounded}{\codestyle{ClusterOrCycleBounded}}
\newcommand{\ClusterOrCycle}{\codestyle{ClusterOrCycle}}
\newcommand{\SimpleCycle}{\codestyle{SimpleCycle}}
\newcommand{\Next}{\codestyle{Next}}
\newcommand{\Sample}{\codestyle{Sample}}
\newcommand{\Dijkstra}{\codestyle{Dijkstra}}
\newcommand{\Preprocess}{\codestyle{Preprocess}}
\newcommand{\HashTable}{\codestyle{HashTable}}
\newcommand{\Heap}{\codestyle{Heap}}
\newcommand{\RelaxNext}{\codestyle{RelaxNext}}

\newcommand{\PreprocessGraph}{\codestyle{Initialize}}
\newcommand{\Route}{\codestyle{Route}}
\newcommand{\TreeRoute}{\codestyle{TreeRoute}}
\newcommand{\N}{\mathbb{N}}
\newcommand{\MinCycle}{\codestyle{MinCycle}}
\newcommand{\Ball}{\codestyle{Ball}}
\newcommand{\DistanceOracle}{\codestyle{TZ-DistanceOracle}}
\newcommand{\SparseOrCycle}{\codestyle{SparseOrCycle}}
\newcommand{\Intersection}{\codestyle{Intersection}}
\newcommand{\CycleAdditive}{\codestyle{CycleAdditive}}
\newcommand{\GenerateSi}{\codestyle{ComputeS}}

\newcommand{\codeNull}{\codestyle{null}}
\newcommand{\codeYes}{\codestyle{Yes}}
\newcommand{\codeNo}{\codestyle{No}}
\newcommand{\codeAnd}{~ \mathrm{and} ~}
\newcommand{\codeOr}{~ \mathrm{or} ~}
\newcommand{\wt}{\ell}
\newcommand{\Cl}{CL}
\newcommand{\CL}{CL}
\newcommand{\cl}{c\ell}

\newcommand{\EQ}{\;=\;}
\newcommand{\GE}{\;\ge\;}
\newcommand{\Ot}{\tilde{O}}
\newcommand{\stactri}{\stackrel\triangle}

\newcommand{\EE}{\mathbb{E}}
\newcommand{\RR}{\mathbb{R}}

\DeclarePairedDelimiter{\ceil}{\lceil}{\rceil}
\DeclarePairedDelimiter{\floor}{\lfloor}{\rfloor}
\DeclarePairedDelimiter{\pair}{\langle}{\rangle}

\title{Faster Algorithms for $(2k-1)$-Stretch Distance Oracles}
\maketitle

\thispagestyle{empty}

\begin{abstract}
The seminal distance oracles of Thorup and Zwick [STOC 2001, JACM 2005] provide optimal stretch/space tradeoffs. However, their $O(mn^{1/k})$ construction time is not optimal, and they posed the question of whether a faster construction time is possible (especially for small $ k$). 
In this paper, we present the first improvement upon their construction algorithm in graphs that are not super sparse, i.e., when $m=\Omega(n^{1+1/k+\eps})$, for any $\eps>0$. 
Moreover, our construction improves upon the $O(n^2)$-time construction of Baswana and Kavitha [FOCS 2006, SICOMP 2010], for every $k > 2$. By achieving the first subquadratic construction for $2<k<6$, we resolve the open problem posed by Wulff-Nilsen [SODA 2012] of whether such subquadratic-time constructions exist. 

Wulff-Nilsen [SODA 2012] targeted nearly linear construction times and presented algorithms running in $\tilde{O}(m+n^{1+f(k)})$ time, which is near-linear whenever the graph density $m$ exceeds the threshold $n^{1+f(k)}$. We obtain improved bounds on $f(k)$ for all $k > 3$, and thus expand the regime of graph densities for which nearly linear construction times are achievable.

In addition, for unweighted graphs, we present several new algorithms for constructing $(2k - 1,\beta)$-oracles that improve upon the results of Baswana, Gaur, Sen, and Upadhyay [ICALP 2008].

Our results are achieved through the development of several new algorithmic tools, which may be of independent interest.
One of our main technical contributions is a hierarchy of parameterized distance oracles, which plays a central role in our fast construction algorithms.
\end{abstract}
\clearpage
\pagenumbering{arabic} 
\newpage

\section{Introduction}

In their seminal work\footnote{Winner of the 20-year Test-of-Time award of STOC 2021.} on distance oracles, Thorup and Zwick~\cite{DBLP:journals/jacm/ThorupZ05} 
presented an algorithm that constructs in $O(kmn^{\frac{1}{k}})$ expected\footnote{Later, Roditty, Thorup, and Zwick~\cite{DBLP:conf/icalp/RodittyTZ05} derandomized the construction.} time, for any integer $k\geq 1$, a $(2k-1)$-stretch distance oracle that uses $O(kn^{1+\frac{1}{k}})$ space and answers distance queries in $O(k)$ time. 
Many different aspects of distance oracles, including fast construction algorithms, have been studied since the introduction of distance oracles \cite{DBLP:conf/focs/MendelN06, ElkinNW16, ElkinP16, DBLP:conf/focs/ElkinS23, DBLP:conf/icalp/Chechik024a, DBLP:conf/soda/Wulff-Nilsen12, DBLP:journals/siamcomp/BaswanaK10,DBLP:conf/icalp/KopelowitzKR24,DBLP:conf/soda/Le23,DBLP:conf/stoc/AbboudBKZ22,DBLP:conf/stoc/AbboudBF23,DBLP:conf/soda/GudmundssonLNS02,DBLP:conf/soda/ChechikCFK17,DBLP:conf/soda/Wulff-Nilsen16}.

Assuming the Erd\H{o}s girth conjecture, the $(2k - 1)$-stretch / $O(kn^{1 + \frac{1}{k}})$-space tradeoff established in~\cite{DBLP:journals/jacm/ThorupZ05} is optimal.
Chechik~\cite{DBLP:conf/stoc/Chechik14, DBLP:conf/stoc/Chechik15} improved the $O(k)$ query time to  $O(1)$. Therefore, the main open problem regarding distance oracles that remains is improving its construction time.  
In the words of Thorup and Zwick~\cite{DBLP:journals/jacm/ThorupZ05}:

\epigraph{``\textit{It
remains an interesting open problem, however, to reduce the preprocessing times of
small stretch oracles.}''}{\textit{Thorup and Zwick~\cite{DBLP:journals/jacm/ThorupZ05}}}

The $O(mn^{\frac{1}{k}})$ construction algorithm of Thorup and Zwick~\cite{DBLP:journals/jacm/ThorupZ05} is super-quadratic (in $n$) whenever $m=\omega(n^{2-\frac{1}{k}})$. 
Baswana and Kavitha~\cite{DBLP:journals/siamcomp/BaswanaK10} addressed this problem and presented an $O(n^2)$ construction-time algorithm, for every $k\ge 2$.
In light of the quadratic time algorithm of Baswana and Kavitha~\cite{DBLP:journals/siamcomp/BaswanaK10}, Wulff-Nilsen~\cite{DBLP:conf/soda/Wulff-Nilsen12} considered the following problem.
\begin{problem}\label{Plm-1}
For graphs with $m=\Theta(n^{2 - \varepsilon})$, for $\eps>0$, for which values of $k \geq 2$, does there exist an $O(n^{2-\delta})$, for $\delta>0$, time algorithm that constructs a $(2k - 1)$-stretch distance oracle with $\Ot(1)$ query time\footnote{$\Ot$ omits a $\mathrm{polylog}(n)$ factor} and $\Ot(kn^{1 + \frac{1}{k}})$ space?
\end{problem}

Wulff-Nilsen~\cite{DBLP:conf/soda/Wulff-Nilsen12} presented an $O(km+kn^{\frac{3}{2}+\frac{2}{k}+O(k^{-2})})$ time construction algorithm, which is truly subquadratic time for every $k\geq 6$. 
Therefore, this resolves~\Cref{Plm-1} for all $k \geq 6$.
Similar to Thorup and Zwick, Wulff-Nilsen suggested focusing on small stretch oracles and asked whether subquadratic time algorithms exist for smaller values of $k$.

In this paper, we provide an almost complete answer to~\Cref{Plm-1} by presenting a truly subquadratic-time construction algorithm for every $k>2$, leaving only the case that $k=2$ open\footnote{We highlight that solving \Cref{Plm-1} for $k = 2$ and achieving a $3$-stretch distance oracle in truly subquadratic time---even without imposing space constraints---would be a significant breakthrough.
Such a result would improve upon the long-standing $\Ot(n^2)$-time $3$-stretch all-pairs shortest paths (APSP) algorithm of Cohen and Zwick~\cite{DBLP:journals/jal/CohenZ01}, and would be a foundational advancement in the development of approximation algorithms for $APSP$.}. We achieve a running time of $O(\max(n^{1+2/k}, m^{1-\frac{1}{k-1}}n^{\frac{2}{k-1}})\log\log{n})$, which for graphs with $m=O(n^{2-\eps})$, is truly subquadratic  time of $O(n^{2-\eps(1-\frac{1}{k-1})})$.
Moreover, this is the first improvement over the $O(mn^{1/k})$ construction time of Thorup and Zwick~\cite{DBLP:journals/jacm/ThorupZ05} in  
graphs with $n^{1+1/k+\eps}<m<n^{2-\eps}$ edges, for any $\eps>0$.

\begin{theorem}\label{T-Construction-no-spanner}\Copy{T-Construction-no-spanner}{
    Let $G=(V, E)$ be a weighted undirected graph and $k\geq 3$ be an integer. There is an $O(\max(n^{1+2/k}, m^{1-\frac{1}{k-1}}n^{\frac{2}{k-1}})\log\log{n})$ time algorithm that constructs a $(2k-1)$-stretch distance oracle that uses $O(n^{1+\frac{1}{k}})$-space and answers distance queries in $O(k\log\log{n})$ time.
    }
\end{theorem}


After breaking the quadratic time barrier for every $k > 2$, we shift our focus to achieving optimal\footnote{$\Omega(m)$ construction time is required due to a simple reduction to $(s,t)$-connectivity.} linear construction time, as formalized in the following problem.

\begin{problem}\label{Plm-2}
For every $k \geq 2$, what is the smallest value $f(k)$ for which a $(2k - 1)$-stretch distance oracle with $\Ot(1)$ query time and space complexity $\Ot(n^{1+1/k})$ can be constructed in $\Ot(m + n^{1 + f(k)})$ time?
\end{problem}

A running time of $\Ot(m + n^{1 + f(k)})$ is nearly linear for graphs where $m = \Omega(n^{1 + f(k)})$, and is therefore \textit{optimal} (up to logarithmic factors) for such graphs. 
In other words,~\Cref{Plm-2} asks: what is the sparsest graph for which it is possible to construct a $(2k - 1)$-stretch distance oracle with $\Ot(1)$ query time and $\tilde{O}(n^{1 + \frac{1}{k}})$ space in almost linear time?

There are two approaches to attack \Cref{Plm-2}. One approach is to prove a lower bound on the value of $f(k)$. 
Jin and Xu~\cite{jin2023removing} 
proved that 
$\Omega(m^{1+\frac{1}{2k-1}-o(1)})$ time
is required for constructing a $(2k-1)$-stretch distance oracle, for graphs in which $m = \Theta(n)$, conditioned on  the 3SUM conjecture. 
Independently, and also conditioned on  the 3SUM conjecture, Abboud, Bringmann, and Fischer~\cite{DBLP:conf/stoc/AbboudBF23} proved that 
$\Omega(m^{1+\frac{1}{2k-1}-o(1)})$ time is required. 
These two results rule out linear-time construction for graphs with $m = \Theta(n)$, which implies that $f(k)>0$.\footnote{Abboud, Bringmann, Khoury, and Zamir~\cite{DBLP:conf/stoc/AbboudBKZ22} previously proved that $\Omega(m^{1+\frac{c}{k}-o(1)})$ construction time is required, for a small constant $c$.} 

In addition, Jin and Xu~\cite{jin2023removing} 
proved that 
$\Omega(m^{1+\frac{1}{4k-3}})$ time is required for constructing a $(2k-1)$-stretch distance oracle, for graphs in which $m = \Theta(n^{1+\frac{1}{4k-4}})$, ruling out linear time construction for graphs with $m = \Theta(n^{1+\frac{1}{4k-4}})$, and therefore showing that $f(k)>\frac{1}{4k-4}$, for every $k>1$. 

Another approach to attack \Cref{Plm-2} is to design a faster construction algorithm which improves the upper bound on $f(k)$.
Wulff-Nilsen~\cite{DBLP:conf/soda/Wulff-Nilsen12} presented two construction algorithms.
The first algorithm with $O(km+kn^{\frac{3}{2}+\frac{2}{k}+O(k^{-2})})$ running time was discussed above, and the second algorithm has $O(\sqrt{k}m+kn^{1+\frac{2\sqrt 6}{\sqrt k}+O(k^{-1})})$ running time.
Combining these two algorithms with the $\Ot(n^2)$-time algorithm of Baswana and Kavitha~\cite{DBLP:journals/siamcomp/BaswanaK10} we get:

$$
f(k)\leq 
\begin{cases}
    1 & 2\leq k < 6\\
    \frac{1}{2}+\frac{2}{k}+O(k^{-2}) & 6 \leq k < 96\\
    \frac{2\sqrt{6}}{\sqrt{k}}+O(k^{-1}) & 96 \leq k
\end{cases}
$$

In this paper, we improve the construction time of the Wulff-Nilsen~\cite{DBLP:conf/soda/Wulff-Nilsen12} algorithms and present two new construction algorithms, as presented in the following two theorems.
\begin{theorem}\label{T-Construction-spanner}\Copy{T-Construction-spanner}{
    Let $G=(V, E)$ be a weighted undirected graph and let $k\geq 3$ be an integer. There is an $O((km+kn^{\frac{3}{2}+\frac{3}{4k-6}})\log\log{n})$ time algorithm that constructs a $(2k-1)$-stretch distance oracle that uses $O(n^{1+\frac{1}{k}})$-space and answers distance queries in $O(k\log\log{n})$ time.}
\end{theorem}

\begin{theorem} \label{T-construction-spanner-ado} \Copy{T-construction-spanner-ado} {
    Let $G=(V, E)$ be a weighted undirected graph and $k\geq 3$ be an integer. There is an $O((\sqrt{k}m+kn^{1+\frac{\sqrt{8k^3-32k+25}-4k+5}{k(k-1)}})\log\log{n})\leq O(\sqrt{k}m\log\log{n}+kn^{1+\frac{2\sqrt{2}}{\sqrt{k}}})$ time algorithm that constructs a $(2k-1)$-stretch distance oracle that uses $O(n^{1+\frac{1}{k}})$-space and answers distance queries in $O(k\log\log{n})$ time.
    }
\end{theorem}

\Cref{T-Construction-spanner} improves the $O(km+kn^{\frac{3}{2}+\frac{2}{k}+O(k^{-2})})$ time of~\cite{DBLP:conf/soda/Wulff-Nilsen12} and \Cref{T-construction-spanner-ado} improves the $O(\sqrt{k}m+kn^{1+\frac{2\sqrt{6}}{\sqrt{k}}+O(k^{-1})})$ time of~\cite{DBLP:conf/soda/Wulff-Nilsen12}.

Combining these two algorithms, we significantly improve upon the current state of the art and obtain the following upper bound on $f(k)$:

$$
f(k)\leq 
\begin{cases}
    1 & k=2,3\\
    \frac12+\frac{3}{4k-6} & 4 \leq k < 16\\
    \frac{2\sqrt{2}}{\sqrt{k}} & 16 \leq k
\end{cases}
$$
In~\Cref{tab:construction}, we summarize our new construction algorithms for weighted undirected graphs and compare them to the previously known results.
In~\Cref{tab:examples} and~\Cref{fig:theorem_weighted}, we compare our new results on the sparsest graphs for which a linear-time construction algorithm exists with the prior work of~\cite{DBLP:journals/siamcomp/BaswanaK10, DBLP:conf/soda/Wulff-Nilsen12}.

Our pursuit of faster construction algorithms leads us to develop several novel tools.
Among them are new \textit{parameterized} distance oracles.
Parameterized distance oracles were first introduced by 
Roditty, Thorup, and Zwick~\cite{DBLP:conf/icalp/RodittyTZ05}. 
Given a set \( S \subseteq V \) as a parameter, they presented an \( O(m|S|^{\frac{1}{k}}) \)-time construction algorithm that builds a \((2k-1)\)-stretch distance oracle with \( O(n|S|^{\frac{1}{k}}) \) space, supporting queries only for vertex pairs \( \langle s,v \rangle \in S \times V \).  

One of our new tools is a  parameterized distance oracle, which is a natural generalization of the 
parameterized distance oracle of~\cite{DBLP:conf/icalp/RodittyTZ05}. Our parameterized distance oracle is constructed in the same running time and uses the same amount of space.
The oracle supports queries for every \( \langle u,v \rangle \) in $V \times V$, rather than in $S\times V$ as~\cite{DBLP:conf/icalp/RodittyTZ05}, at the cost of adding \( 2\min(h(u),h(v)) \) to the estimate, where \( h(x) \) is the distance from \( x \) to \( S \).
Notice that if \( u \in S \), then \( h(u) = 0 \), and the stretch is  \((2k-1)\) as in~\cite{DBLP:conf/icalp/RodittyTZ05}.

The main technical contribution of this paper is a new algorithmic tool called hierarchical distance oracles, which is composed of a hierarchy of parameterized distance oracles. This construction consists of $\log \log n$ such oracles, each built on a distinct subgraph of $G$ and parameterized by a different set $S$.

Using our new hierarchical distance oracle we obtain our $O(\max(n^{1+2/k}, m^{1-\frac{1}{k-1}}n^{\frac{2}{k-1}}))$-time algorithm for constructing a $(2k-1)$-stretch distance oracle that uses $O(n^{1+\frac{1}{k}})$ space, mentioned above. We remark that this improves upon the $O(mn^{1/k})$-construction time of Thorup and Zwick~\cite{DBLP:journals/jacm/ThorupZ05}, for graphs with $m=\Omega(n^{1+\frac{1}{k}+\eps})$ edges.

In addition, we consider unweighted graphs.
An estimation $\hat{d}(u,v)$ of $d(u,v)$ is an 
$(\alpha,\beta)$-approximation if $d(u,v)\leq \hat{d}(u,v) \leq \alpha d(u,v) + \beta$.
$(\alpha,\beta)$-approximations were extensively studied in the context of graph spanners, emulators, and distance oracles. (For more details see for example~\cite{DBLP:journals/siamcomp/ElkinP04, DBLP:conf/soda/ThorupZ06, DBLP:journals/talg/BaswanaKMP10, DBLP:conf/soda/Chechik13, VassilevskaSp15, DBLP:conf/icalp/Parter14, DBLP:journals/jacm/AbboudB17, DBLP:journals/siamcomp/AbboudBP18, DBLP:journals/corr/abs-1201-2703, DBLP:conf/esa/Agarwal14, DBLP:conf/icalp/Chechik024a, DBLP:journals/theoretics/BiloCCC0KS24, DBLP:conf/icalp/KopelowitzKR24, DBLP:conf/soda/AgarwalG13}).

Baswana, Gaur, Sen, and Upadhyay~\cite{DBLP:conf/icalp/BaswanaGSU08} considered estimations with an additive error, to obtain faster construction-time algorithms for distance oracles. They studied the following problem.
\begin{problem}\label{Plm-3}
In unweighted undirected graphs, for every $k\geq 2$ and $\beta\geq 0$, what is the smallest value $f(k,\beta)$,
for which a $(2k-1,\beta)$-approximation distance oracle with $\Ot(1)$ query time that uses $\Ot(n^{1+\frac{1}{k}})$ space can be constructed in $\Ot(m+n^{1+f(k,\beta)})$ time?
\end{problem}
Baswana, Gaur,  Sen, and Upadhyay~\cite{DBLP:conf/icalp/BaswanaGSU08}
presented an algorithm for constructing a $(2k-1,2)$-approximation distance oracle that uses $O(kn^{1+\frac{1}{k}})$ space, for $k\geq 3$. 
The algorithm runs in $O(m+kn^{\frac{3}{2}+\frac{1}{k} + \frac{1}{k(2k-2)}})$ time. 
For $k=2$, they presented an algorithm that runs in $\Ot(m+n^{23/12})$ time and constructs a $(3,14)$-approximation distance oracle that uses   $O(n^{1.5})$ space. 

Combining these algorithms with the construction of Wulff-Nilsen~\cite{DBLP:conf/soda/Wulff-Nilsen12} we get that the current state of the art for~\Cref{Plm-3} is:
$$
f(k,\beta)\leq
\begin{cases}
    \frac{1}{2}+\frac{1}{k} + \frac{1}{k(2k-2)} & 3 \leq k < 96,  \beta =2\\
    \frac{2\sqrt{6}}{\sqrt{k}}+O(k^{-1}) & 96 \leq k,  \beta =0
\end{cases}
$$

In this paper, we improve the results of Baswana, Gaur, Sen, and Upadhyay~\cite{DBLP:conf/icalp/BaswanaGSU08} and present three new construction algorithms. The first construction algorithm runs in $\Ot(km+kn^{\frac{3}{2}+\frac{1}{k-1}-\frac{1}{4k-6}})$ and constructs a $(2k-1,2)$-approximation distance oracle, improving the $\Ot(m+kn^{\frac{3}{2}+\frac{1}{k} + \frac{1}{k(2k-2)}})$ construction time of~\cite{DBLP:conf/icalp/BaswanaGSU08}, for every $k>2$.
The second algorithm runs in $\Ot(km+kn^{\frac{3}{2}+\frac{1}{2k}+\frac{3.5k-4.5}{k(4k^2-8k+3)}})$ time and constructs a $(2k-1,2k-2)$-approximation distance oracle.
The third algorithm runs in $\Ot(\sqrt{k}m+kn^{1+\frac{2}{\sqrt{k}}})$ time and constructs a $(2k-1,2k-1)$-approximation oracle.
Combining these three algorithms, we significantly improve upon the current state of the art and obtain the following values for $f(k, \beta)$:

$$
f(k, \beta) \leq
\begin{cases}
    11/12 & k =2, \beta =14\\
    \frac{1}{2}+\frac{1}{k-1}-\frac{1}{4k-6} & 3 \leq k < 13,  \beta=2\\
    \frac{1}{2}+\frac{1}{2k}+\frac{3.5k-4.5}{k(4k^2-8k+3)} & 3 \leq k < 13,  \beta=2k-2\\
    \frac{2}{\sqrt{k}} & 13 \leq k, \beta=2k-1
\end{cases}
$$
In addition, we present an $O(mn^\eps+n^{1+1/k+\eps})$ time algorithm for constructing a $((2k-1)(1+\eps),\beta)$-distance oracle, improving the stretch of the $O(m+n^{1+1/k+\eps})$ time algorithm of~\cite{DBLP:conf/soda/Wulff-Nilsen12} that constructs an $O(k)$-distance oracle.

\ifC \else
In~\Cref{tab:construction_unweighted}, we summarize our new construction algorithms for unweighted undirected graphs and compare them to the previously known results.
In~\Cref{tab:examples-unweighted} and~\Cref{fig:theorem_unweighted}, we compare the previous results of~\cite{DBLP:conf/icalp/BaswanaGSU08} to our new results regarding the sparsest graphs for which there exists a linear time construction algorithm. \fi

The rest of this paper is organized as follows. In the next section, we present some necessary preliminaries. In~\Cref{S-OV}, we overview our main results and techniques. In~\Cref{S-NT} we present our new tools, and specifically in~\Cref{S-hado} we present our new hierarchical distance oracle. In~\Cref{S-Construction-weighted} we present our construction algorithms for weighted graphs\ifC\else, and in~\Cref{S-Construction-unweighted} we present our construction algorithms for unweighted graphs\fi.
For convenience, in~\Cref{S-comparison} we include tables and graphs that summarize all the theorems.

\section{Preliminaries}\label{S-Prel}
Let $G=(V,E)$ be an undirected graph with $n=|V|$ vertices and $m=|E|$ edges. Throughout the paper, we consider unweighted and weighted graphs with non-negative real edge weights. 
Let $u,v \in V$, we denote $d_G(u,v)$ as the distance between $u$ and $v$ in $G$.
Let $P(u,v)$ be the shortest path between $u$ and $v$.
Let $N(u)\subseteq V$ be the neighbors of $u$ and let $deg(u)=|N(u)|$ be the degree of $u$.

Let $S\subseteq V$.
The distance $d(u,S)$ between $u$ and $S$ is the distance between $u$ and the closest vertex to $u$ from $S$, that is, 
$d(u,S)= \min_{s\in S}(d(u,s))$. We denote $d(u, S)$ with $h_S(u)$, and when $S$ is clear from the context, we write $h(u)$.
Let $p_S(u)=\arg \min_{s\in S}(d(u,s))$ (ties are broken in favor of the vertex with a smaller identifier). 
Let $B_S(u)=\{v\in V \mid d_G(u,v) < d_G(u,S) \}$.
Let $E_S(u)$ be the set of edges incident to $u$ whose weight is less than $h_S(u)$, that is $E_S(u)=\{(u,v)\in E \mid \ell(u,v)< h_S(u)\}$.
Let $E_S=\bigcup_{v\in V}E_S(v)$ and $G_{S}=(V,E_S)$. See \Cref{fig:G_S} for an illustration. 

The following lemmas are standard results and are included here for completeness.
\begin{lemma}\label{L-Construct-S}
    Let $G = (V, E)$ be a weighted undirected graph. Given a non-empty
    subset $S \subseteq V$, $p_S(u)$ and $h_S(u)$ can be computed in $O(m + n\log n)$ time for all vertices $u \in V$.
\end{lemma}
\begin{proof}
    Let $G'=(V\cup \{x\}, E\cup \{(x,s,0)\mid s\in S\})$. 
    Running Dijkstra from $x$ in $G'$ computes $d_{G'}(x,u)=d_G(S,u)=h_{S}(u)$, and identifies $p_S(u)$, the second vertex in $P_{G'}(x,u)$, for every $u\in V$.
\end{proof}

\begin{lemma}\label{L-Create-G_S}
    Let $G = (V, E)$ be a weighted undirected graph and
    let $S\subseteq V$ be a non-empty
    set, the graph $G_{S}$ can be computed in $O(m + n \log n)$ time.
\end{lemma}
\begin{proof}
    Using~\Cref{L-Construct-S}, we obtain $h_S(u)$ for every $u\in V$ in $O(m + n\log n)$ time. In $O(m)$ time we iterate every edge $(u,v)$ and remove it if $\ell(u,v) \ge \max(h_S(u), h_S(v))$.
\end{proof}

Baswana and Kavitha~\cite{DBLP:journals/siamcomp/BaswanaK10} proved the following useful lemma on the graph $G_S=(V, E_S)$.

\begin{lemma}[\cite{DBLP:journals/siamcomp/BaswanaK10}]\label{L-G_S-size}
Let $G = (V, E)$ be a weighted undirected graph and let $S\subseteq V$ be a non-empty set. 
For any two vertices $u,v\in V$, if $v\in B_S(u)$ then
$d_{G_{S}}(u, v) = d_G(u, v)$. If $S$ is obtained by picking each vertex independently with probability $p$, then $E_S$ has
expected size $O(n/p)$.
\end{lemma}
\begin{figure}
    \centering
    \tikzset{every picture/.style={line width=0.75pt}} 
    \input{figures/mathcha_G_S}

    \caption{Illustration of $G_S$, the distance is the Euclidean distance.}
    \label{fig:G_S}
\end{figure}When 

Throughout the paper, we occasionally distinguish between the case that $P(u,v)\not\subseteq G_{S}$ and the case that $P(u,v)\subseteq G_{S}$. Next, we prove a useful bound on $h_S(u)$ and $h_S(v)$ when $P(u,v)\not\subseteq G_{S}$ (see~\Cref{fig:no-intersect}).
\begin{lemma}\label{L-No-Intersect-Weighted}
    Let $G$ be a weighted undirected graph.
    If $P(u,v)\not\subseteq G_{S}$, then $\max(h(u),h(v))\le d_G(u,v)$
\end{lemma}
\begin{proof}
Since $P(u,v)\not\subseteq G_{S}$ there exists an edge $(w,z)\in P(u,v)$ such that $(w,z)\notin G_{S}$.  
Wlog, we assume that $d(u,w)\le d(u,z)$.
By the definition of $G_{S}$, the fact that $(w,z)\notin G_{S}$ implies that
\(
    h_{S}(w)\le \ell(w,z).
\)
Since $(w,z)\in P(w,v)$, we have
\(
    \ell(w,z)\le d(w,v).
\)
Moreover, by the triangle inequality for $h_S(u)$, we get that
\(
    h_{S}(u)\le d(u,w)+h_{S}(w).
\)
By combining these inequalities, we get:
\[
    h_{S}(u) 
    \;\le\; d(u,w)+h_{S}(w) 
    \;\le\; d(u,w)+\ell(w,z) 
    \;\le\; d(u,w)+d(w,v)
    \;\le\; d(u,v)
\]
Where the equality follows from the fact that $w\in P(u,v)$, as required.
Using symmetric arguments, we can also get that $h(v) \le d(u,v)$. 
\end{proof}

\colorlet{cpath}{black}            
\definecolor{cpruned}{RGB}{208,2,27} 
\colorlet{cset}{blue!60!black}     

\begin{figure}[t]
\centering
\begin{tikzpicture}[
    vert/.style  = {circle, draw=black, fill=white, minimum size=6pt, inner sep=0pt, line width=0.8pt},
    svert/.style = {circle, draw=cset, fill=cset, minimum size=6pt, inner sep=0pt},
    sub/.style   = {cpath, line width=1pt,
                    decorate, decoration={snake, amplitude=1.2pt, segment length=7pt,
                    pre length=4pt, post length=4pt}},
    branch/.style= {cset, line width=1pt},
]

\coordinate (u) at (0,0);
\coordinate (w) at (2.2,0);
\coordinate (z) at (3.4,0);
\coordinate (v) at (5.6,0);
\coordinate (s) at (2.2,1.35);

\draw[sub] (u) -- (w);
\draw[sub] (z) -- (v);
\draw[cpruned, dashed, line width=1.4pt] (w) -- (z);
\draw[branch] (w) -- (s);

\node[vert] at (u) {}; \node[vert] at (w) {}; \node[vert] at (z) {}; \node[vert] at (v) {};
\node[svert] at (s) {};

\node[font=\small, anchor=north] at ($(u)+(0,-0.10)$) {$u$};
\node[font=\small, anchor=north] at ($(w)+(0,-0.10)$) {$w$};
\node[font=\small, anchor=north] at ($(z)+(0,-0.10)$) {$z$};
\node[font=\small, anchor=north] at ($(v)+(0,-0.10)$) {$v$};
\node[font=\small, cset, anchor=south] at ($(s)+(0,0.10)$) {$p_S(w)$};
\node[font=\scriptsize, cset, anchor=east] at (2.07,0.66) {$h_S(w)$};
\node[cpruned, font=\scriptsize] at (2.8,0.30)  {$\wt(w,z)$};
\node[cpruned, font=\scriptsize] at (2.8,-0.78) {$(w,z)\notin G_S$};

\end{tikzpicture}
\caption{
Illustration of \Cref{L-No-Intersect-Weighted}: if $P(u,v)\not\subseteq G_S$ then $h_S(u)\le d(u,v)$. Let $(w,z)\not\in G_S$, then $h_S(w)\le\wt(w,z)$.
}
\label{fig:no-intersect}
\end{figure}

For a real value $\alpha$, an $\alpha$-stretch spanner of $G$ is a subgraph $H=(V,E_H)$ of $G$ such that for every $u,v\in V$, $d_G(u,v) \le d_H(u,v) \le \alpha \cdot d_G(u,v)$.
An $(\alpha, \beta)$-approximation spanner of $G$ is a subgraph such that for every $u,v\in V$, $d_G(u,v) \le d_H(u,v) \le \alpha \cdot d_G(u,v) + \beta$.
In the next lemma, we describe the spanner construction that we use in our algorithms.
\begin{lemma}[\cite{DBLP:journals/rsa/BaswanaS07}]\label{L-Spanner-2k-1}
For any integer $k$, for a weighted graph $G$, a $(2k-1)$-spanner with $O(n^{1+\frac{1}{k}})$ edges can be computed in $O(km)$ time.
\end{lemma}

In the following lemma, we summarize the properties of the distance oracle construction algorithm as presented in~\cite{DBLP:journals/jacm/ThorupZ05}.
\begin{lemma}[\cite{DBLP:journals/jacm/ThorupZ05}]\label{L-Regular-ADO-Construction}
    There is an $O(kn^{1+\frac{1}{k}})$-space $(2k-1)$-stretch distance oracle with $O(k)$ query time. The distance oracle is constructed in $O(mn^{\frac{1}{k}})$ expected time. 
\end{lemma}
In our construction algorithms, we refer to the distance oracle of Thorup and Zwick~\cite{DBLP:journals/jacm/ThorupZ05}, presented in~\Cref{L-Regular-ADO-Construction}, as $\ADO(G,k)$. 
We denote a query to $\ADO(G, k)$ by\\ $\ADO(G, k)\Query(u, v)$, where $u$ and $v$ are the queried vertices. We adopt this notation for queries to all other distance oracles developed in this paper.

\section{Overview}\label{S-OV}
This section provides an overview of our construction algorithms and the key new tools we developed. To put our contribution in context, we begin by describing the framework underlying both construction algorithms of Wulff-Nilsen~\cite{DBLP:conf/soda/Wulff-Nilsen12}. This framework consists of the following three steps:

\begin{enumerate}[label=(Step \arabic*), leftmargin=*]
\item Sample a set $S$ and create an oracle that upon a query on $u,v\in V$ produces $\hat{d}_{G_S}(u,v)$.\label{i-1}
\item Compute a spanner $H$.\label{i-2}
\item Create an oracle that upon a query on $s_1,s_2\in S$ produces $\hat{d}_{H}(s_1,s_2)$.
\label{i-3}
\end{enumerate}

The distance oracle query of this framework returns $\min(\hat{d}_{G_S}(u,v), h_S(u) + h_S(v) + \hat{d}_H(p(u),p(v)))$.
Using this framework, Wulff-Nilsen  obtained   
an $O(\sqrt{k}m+kn^{1+\frac{2\sqrt{6}}{\sqrt{k}}+O(k^{-1})})$ construction-time algorithm as follows. In~\ref{i-1} $\ADO(G_S,k)$ is constructed. In~\ref{i-2} a $(2k'-1)$-spanner $H$ from~\Cref{L-Spanner-2k-1} is computed. In~\ref{i-3}  $\ADO(H,k'')$ is constructed.

Next, we discuss the new main techniques that allow us to improve~\ref{i-1}. 
Notice that there is a tradeoff between $|S|$ and $|E_S|$. 
When $|E_S|$ is small,~\ref{i-1} becomes faster and $|S|$ increases, resulting in a slower~\ref{i-3}. 
Given this tradeoff, the key to obtaining faster construction algorithms lies in the following optimization problem:

\begin{problem}\label{Plm-4}
For every $0 \le x_0 \le 1$, what is the densest graph for which a $(2k - 1)$-stretch distance oracle with $\Ot(1)$ query time can be constructed in $\Ot(n^{1 + x_0 + \frac{1}{k}})$ time?
\end{problem}

Thorup and Zwick's~\cite{DBLP:journals/jacm/ThorupZ05}
construction algorithm, which is used by 
Wulff-Nilsen for ~\ref{i-1}, solves~\Cref{Plm-4} for graphs with  $O(n^{1+x_0})$ edges.
We improve the construction algorithm of~\cite{DBLP:journals/jacm/ThorupZ05}, for every $x_0 > \frac{1}{k}$ and solve~\Cref{Plm-4} for graphs with  $O(n^{1+\frac{k-1}{k-2}x_0-\frac{1}{k(k-2)}})$ edges.

To obtain this improvement, we develop a new hierarchy of parameterized distance oracles, which we denote with $\hADO$ (see~\Cref{T-Construct-to-T}).
The hierarchical distance oracle uses our new parameterized distance oracle, which we denote with $\ADO_P$ (see~\Cref{L-ADO-G_S}).


To understand how $\hADO$ works, we first describe
parameterized distance oracles.
Let $S\subseteq V$. Roditty, Thorup, and Zwick~\cite{DBLP:conf/icalp/RodittyTZ05} presented an $O(m|S|^{\frac{1}{k}})$-time construction algorithm for constructing a $(2k-1)$-stretch distance oracle that uses $O(n|S|^{\frac{1}{k}})$ space and supports queries  for vertex pairs $\langle s, v \rangle \in S \times V$.
Our new parameterized distance oracle, $\ADO_P$,  is a generalization of the parameterized distance oracle of~\cite{DBLP:conf/icalp/RodittyTZ05}.  $\ADO_P$ is
constructed in the same running time, uses the same space and query time, and supports queries for vertex pairs $\langle u, v \rangle \in V \times V$ rather than $\langle s, v \rangle \in S \times V$, and returns an estimate satisfying:
$$\ADO_P(G,k,S)\Query(u,v) \le 2\min(h_S(u),h_S(v))+(2k-1)d(u,v)$$
Notice that if $u\in S$ then $h_S(u)=0$ and $\ADO_P(G,k,S)\Query(u,v) \le (2k-1)d(u,v)$, as in~\cite{DBLP:conf/icalp/RodittyTZ05}. (For more details, see~\Cref{S-Parametrized-oracle}.)

The goal of the $\hADO$ is to compute distance oracles on denser graphs. To this end, we use a hierarchy of graphs that starts with sparse graphs and ends with denser graphs. Let $u,v\in V$ and let $G_{S_0}$, for some set $S_0$, be the first graph in the hierarchy.  The main observation is that if  
we already know that $P(u,v) \not\subseteq G_{S_0}$
then, by~\Cref{L-No-Intersect-Weighted}, we have $h_{S_0}(u), h_{S_0}(v) \le d(u,v)$. 

This  allows us to consider the denser graph  $G_{S_1}$, while maintaining the same stretch and running time by computing $\ADO_P(G_{S_1}, k - 1, 
S_0)$. 
This solves the case in which $P(u,v) \subseteq G_{S_1}$. If $P(u,v) \not\subseteq G_{S_1}$ by~\Cref{L-No-Intersect-Weighted} we have that $h_{S_1}(u), h_{S_1}(v) \le d(u,v)$ and we continue the process by computing $\ADO_P(G_{S_2}, k - 1, S_1)$.  We continue this process for $\log\log n$ levels, until two consecutive steps yield the same asymptotic density.

More formally,
the hierarchical distance oracle $\hADO$ works as follows (see~\Cref{fig:hierarchy}).
Let $S_0$ be a set of size $n^{1-x_0}$, and construct $\ADO(G_{S_0},k)$.
Let $t=\ceil{\log\log{n}}$.
For every $1\leq i \leq t$, construct $\ADO_P(G_{S_i}, k-1, S_{i-1})$, (roughly speaking, we use $k-1$ since the first level is the ``$S$'' set, this is required to get the desired stretch guarantee), where $S_i$ is a random set of size $n^{1-x_i}$ and $x_i=x_0-\frac{1}{k(k-1)}+\frac{x_{i-1}}{k-1}$.\footnote{Notice that $x_{i+1} - x_{i} = \frac{x_{i} - x_{i-1}}{k-1}\Rightarrow x_{i+1}>x_i$. Therefore $|S_{i+1}|<|S_i|$ and $|E_{S_{i+1}}|>|E_{S_{i}}|$.} The query of $\hADO$ returns $\min(\ADO(G_{S_0},k)\Query(u,v),\min_{1 \leq i \leq t}(\ADO_P(G_{S_i},k-1,S_{i-1})\Query(u,v)))$.

The construction time of $\hADO$ is $\Ot(n^{1+x_0+\frac{1}{k}})$. 
To see this notice that constructing $\ADO_P(G_{S_i},k-1,S_{i-1})$ takes $O(|E_{S_i}|\times |S_{i-1}|^{\frac{1}{k-1}})=O(kn^{1+x_i+\frac{1-x_{i-1}}{k-1}})$ time, and since $x_i=x_0-\frac{1}{k(k-1)}+\frac{x_{i-1}}{k-1}$ we get
$O(kn^{1+x_0-\frac{1}{k(k-1)}+\frac{x_{i-1}}{k-1}+\frac{1-x_{i-1}}{k-1}})=O(kn^{1+x_0+\frac{1}{k}})$.

For every $u,v$, where $P(u,v)\subseteq G_{S_t}$, we have that $\hADO\Query(u,v)\leq (2k-1)d(u,v)$ (see~\Cref{fig:hierarchy}).
To see that, let $j$ be the first index for which $P(u,v)\subseteq G_{S_j}$. Since $P(u,v)\not\subseteq G_{S_{j-1}}$, it follows from~\Cref{L-No-Intersect-Weighted} that 
$\max(h_{S_{j-1}}(u), h_{S_{j-1}}(v)) \le d(u,v)$,
thus, since $\ADO_P$ was constructed with parameter $k-1$ we get $\ADO_P(G_{S_j},k-1,S_{j-1})\Query(u,v) \leq 2\min(h_{S_{j-1}}(u),h_{S_{j-1}}(v))+(2(k-1)-1)d(u,v)\leq (2k-1)d(u,v)$. (See~\Cref{L-Construct-to-T-correction} for a complete  proof.)

Our new hierarchical distance oracle, $\hADO$, obtains a truly subquadratic construction algorithm for every $k\ge 3$, even without using the three-step framework presented above. 
Roughly speaking, the construction algorithm works as follows.
Let $x_0$ be the solution for the equation $n^{1+\frac{k-1}{k-2}x_0-\frac{1}{k(k-2)}}=m$. One can verify that $x_0 = \frac{(\log_n m)(k^2 - 2k) - (k^2 - 2k - 1)}{k(k - 1)}$.

We construct $\hADO(G,k,x_0)$. 
In the hierarchy, we have that $|E_{S_t}|=n^{1+\frac{k-1}{k-2}x_0-\frac{1}{k(k-2)}}=m$. Since $|E_{S_t}|=m$ we can replace $G_{S_t}$ with $G$ without affecting the construction time. 
Since $P(u,v)\subseteq G$ for every $u,v\in V$, we obtain a $(2k-1)$-stretch distance oracle.
The construction time of this distance oracle is $\Ot(n^{1+x_0+\frac{1}{k}})=\Ot(m^{1-\frac{1}{k-1}}n^{\frac{2}{k-1}})$ (See~\Cref{L-Construction-time-no-spanner}).
We remark that this is truly subquadratic since for $m=n^{2-\eps}$ we get a running time of $\Ot(n^{(2-\eps)(1-\frac{1}{k-1})}n^{\frac{2}{k-1}})=\Ot(n^{2-\eps(1-\frac{1}{k-1})})$, which is truly subquadratic for every $k>2$.

Next, we overview our improvement in the regime of linear construction-time algorithms.
To improve the previous results, we use the framework described above and improve each step of the three-step framework as follows:
In~\ref{i-1} we use $\hADO(G,k,x_0)$ as opposed to $\ADO(G,k)$ as in~\cite{DBLP:conf/soda/Wulff-Nilsen12}.
This enables us to estimate distances in $G_{S_t}$ instead of $G_S$.
In~\ref{i-2} we 
compute a $(2k'-1)$-spanner $H$ as in~\cite{DBLP:conf/soda/Wulff-Nilsen12}. However, we augment $H$ with edges between $u$ and $p(u)$ of weight $h(u)$ for every $u\in V$.
This enables us to bound $d_H(p(u),p(v))$ with $h(u)+h(v)+(2k'-1)d(u,v)$ instead of $ (2k'-1)(h(u)+h(v)+d(u,v))$ as in~\cite{DBLP:conf/soda/Wulff-Nilsen12}.

For~\ref{i-3} we develop an additional parameterized distance oracle, $\ADO_{P'}$ (see~\Cref{L-ADO-S-Only}), that is constructed in $O(m|S|^{\frac{1}{k}})$ time, the same running time as  $\ADO_P$, but uses $O(|S|^{1+\frac{1}{k}})$ space which is significantly less space than the $O(n|S|^{\frac{1}{k}})$ space required by $\ADO_P$. The space reduction comes at a cost of supporting queries only for vertex pairs in $S\times S$.
In~\ref{i-3} we use $\ADO_{P'}(H,k-1,{S_t})$ instead of $\ADO(H,k)$ as in~\cite{DBLP:conf/soda/Wulff-Nilsen12}.

These improvements allow us to reduce the $O(\sqrt{k}m+kn^{1+\frac{2\sqrt{6}}{\sqrt{k}}+O(k^{-1})})$ construction time of~\cite{DBLP:conf/soda/Wulff-Nilsen12} to $O(\sqrt{k}m+kn^{1+\frac{2\sqrt{2}}{\sqrt{k}}})$ and the $\Ot(km+kn^{\frac{3}{2}+\frac{2}{k}+O(k^{-2})})$ construction time of~\cite{DBLP:conf/soda/Wulff-Nilsen12} to $\Ot(km+kn^{\frac{3}{2}+\frac{3}{4k-6}})$. (See \Cref{S-construction-spanner-ado,S-construction-spanner}.) 

By applying our new techniques to unweighted graphs and allowing a small additive error, we can further improve the construction time.
In particular, among our results for unweighted graphs, we obtain 
a construction algorithm that in $O(km+kn^{\frac{3}{2}+\frac{0.75}{k}+O(k^{-2})})$ time constructs a $(2k-1,2)$-approximation distance oracle, improving the $O(km+kn^{\frac{3}{2}+\frac{1}{k}+O(k^{-2})})$ of Baswana, Gaur, Sen, and  Upadhyay~\cite{DBLP:conf/icalp/BaswanaGSU08}, for every $k>2$.

We remark that all of our construction algorithms can be de-randomized using the methods presented by Roditty, Thorup, and Zwick~\cite{DBLP:conf/icalp/RodittyTZ05} at a logarithmic cost.

\section{Tools}\label{S-NT}
In this section, we develop our new techniques that we later use in our fast construction algorithms presented in Sections~\ref{S-Construction-weighted} and\ifC\else~\ref{S-Construction-unweighted}\fi. 
We begin by proving several useful properties of the classical oracle of Thorup and Zwick~\cite{DBLP:journals/jacm/ThorupZ05} that are used by our parameterized distance oracles, presented right afterwards. 
The main new technical contribution, our new hierarchical distance oracle, is presented in the last part of this section, see~\Cref{S-hado}.

\subsection{Properties of the classical distance oracle of~\cite{DBLP:journals/jacm/ThorupZ05}}
In this section, we discuss the internals and some properties of the Thorup and Zwick distance oracle~\cite{DBLP:journals/jacm/ThorupZ05}, which are then applied to our parameterized distance oracles. 

Given an integer parameter $k\geq 2$, we can compute in $O(n)$ time vertex sets $A_1,\ldots, A_{k-1}$, such that $|A_i|=O(n^{1-i/k})$, for every $i\in [1,k-1]$. Let $A_0=V$ and let $A_k=\emptyset$. Let $p_i(u)$ be the closest vertex to $u$ from the set $A_i$, for every $i\in [0,k]$.
For every $i\in [0,k-1]$, let the \textit{bunch} $B_i(u)$ be the set $\{ v \in A_i  \mid   d(u,v)<d(u,p_{i+1}(u))\}$. Let $B(u)=\bigcup_{i=0}^{k-1} B_i(u)$. Let $w\in A_{i}\setminus A_{i+1}$ and let  $C(w)=\{ u \mid w\in B_i(u)\}$. 

\begin{lemma}[\cite{DBLP:journals/jacm/ThorupZ05}]\label{L-TZ-Size}
The size of $B(u)$ is $O(kn^{\frac{1}{k}})$, whp (with high probability\footnote{
We say that an event occurs \emph{with high probability (whp)} if it occurs with probability at least $1 - \frac{1}{n^c}$ for some constant $c > 0$, where $n$ is the size of the input (e.g., the number of vertices or edges in the graph).
}) and the cost of computing $B(u)$, for every $u\in V$, is $O(mn^{\frac{1}{k}})$ whp.
\end{lemma}
Next, for completeness, we describe the query of~\cite{DBLP:journals/jacm/ThorupZ05} ($\ADOQuery$). Later, we generalize its correctness proofs, and this generalization helps us to achieve the oracles of~\Cref{S-Parametrized-oracle}.

The procedure $\ADOQuery$ gets as an input $u,v\in V$, and using their bunches gives an estimation $\hat{d}(u,v)$. Let $h_i(u)$ be $h_{A_i}(u)$.
For $i\in [k]$, if $p_i(u)\in B_i(v)$ return $h_i(u)+d(p_i(u),v)$, otherwise swap $u$ and $v$ and continue iterating the for-loop.
A pseudo-code for $\ADOQuery$ is presented in~\Cref{A-ADO-Query}.
\begin{algorithm2e}[t] 
\caption{$\ADOQuery(u,v)$ - The query of the Thorup-Zwick\cite{DBLP:journals/jacm/ThorupZ05} distance oracle}\label{A-ADO-Query}
\For{$i\in [k]$} {
    \If{$p_{i}(u)\in B_i(v)$}{\Return $h_{i}(u)+d(p_{i}(u),v)$}
    $(u,v) \gets (v,u)$
}
\end{algorithm2e}
Let $\hat{d}(u,v)$ be $\min(\ADOQuery(u,v), \ADOQuery(v,u))$.
The following lemmas prove properties of $\hat{d}(u,v)$, which generalizes the proof of~\cite{DBLP:journals/jacm/ThorupZ05}.
\begin{lemma} \label{L-ADO.Query-Correctness-1}
    $\hat{d}(u,v) \le 2\min(h_{k-1}(u), h_{k-1}(v)) + d(u,v)$
\end{lemma}
\begin{proof}
    Since $A_k=\emptyset$, we have that $p_{k-1}(u)\in B_{k-1}(v)$.
    Thus, $\hat{d}(u,v) \le h_{k-1}(u) + d(p_{k-1}(u),v) \le 2h_{k-1}(u) + d(u,v)$, where the last inequality follows from the triangle inequality. 
    Using symmetrical arguments, we have that  $\hat{d}(u,v) \le 2h_{k-1}(v) + d(u,v)$, as required.
\end{proof}
\begin{lemma}\label{L-ADO.Query-Correctness-2}
    For every integer $1 \le i \le k$, one of the following holds.
    \begin{itemize}
        \item $\min(h_i(u),h_i(v)) \le \min(h_{i-1}(u),h_{i-1}(v)) + d(u,v)$
        \item $\hat{d}(u,v) \le 2\min(h_{i-1}(u),h_{i-1}(v)) + d(u,v)$
    \end{itemize}
\end{lemma}
\begin{proof}
    Wlog (without loss of generality), assume that $h_{i-1}(u) \le h_{i-1}(v)$. 
    We divide the proof into two cases.
    The case where $p_{i-1}(u)\in B_{i-1}(v)$, and the case where $p_{i-1}(u)\notin B_{i-1}(v)$.
    Consider the case where $p_{i-1}(u)\in B_{i-1}(v)$. In this case, 
    the value of $d(p_{i-1}(u),v)$ is saved in the distance oracle. From the triangle inequality it follows that $d(p_{i-1}(u),v)\le h_{i-1}(u) + d(u,v)$. Therefore, we have that $\hat{d}(u,v) \le h_{i-1}(u)+d(p_{i-1}(u),v) \le 2h_{i-1}(u)+d(u,v)$.
    Since $h_{i-1}(u) \le h_{i-1}(v)$ we get that $\hat{d}(u,v) \le 2\min(h_{i-1}(u),h_{i-1}(v)) + d(u,v)$,
    as required.
    
    Consider now the case that $p_{i-1}(u)\notin B_{i-1}(v)$. From the definition of $B_{i-1}(v)$ it follows that $h_i(v) \le d(v,p_{i-1}(u)) \le d(u,v) + h_{i-1}(u)$. Since $h_{i-1}(u) \le h_{i-1}(v)$, we get that $h_i(v) \le \min(h_{i-1}(u),h_{i-1}(v)) + d(u,v)$, and thus $\min(h_i(u),h_i(v)) \le \min(h_{i-1}(u),h_{i-1}(v)) + d(u,v)$, as required.
\end{proof}

\subsection{Parameterized distance oracles}\label{S-Parametrized-oracle}
A \textit{parameterized distance oracle} is constructed with respect to a special input set $S\subseteq V$. 
Roditty et al.~\cite{DBLP:conf/icalp/RodittyTZ05} presented an algorithm that constructs such a $(2k-1)$-distance oracle in 
$O(m|S|^{\frac{1}{k}})$ time. The distance oracle uses $O(n|S|^{\frac{1}{k}})$ space and answers distance queries only for vertex pairs $\langle s,v \rangle\in S \times V$ in $O(k)$ time.

In this section, we present two new parameterized distance oracles used in our fast constructions. 
In the first parameterized distance oracle, the query supports vertex pairs from $V\times V$. In the second parameterized distance oracle, the query supports vertex pairs only from $S\times S$.

First, we generalize the parameterized distance oracle of~\cite{DBLP:conf/icalp/RodittyTZ05}, and extend the query to support vertex pairs from $V\times V$, as follows.
For every $u,v\in V$ the query returns an estimation $\hat{d}(u,v)\leq 2\min(h_S(u),h_S(v))+(2k-1)d(u,v)$. 
Notice that when $u\in S$ we have  $h_S(u)=0$ thus $\min(h_S(u),h_S(v))=0$, and $\hat{d}(u,v)\le (2k-1)d(u,v)$,  as in the distance oracle of~\cite{DBLP:conf/icalp/RodittyTZ05}.
Formally:
\begin{lemma} \label{L-ADO-G_S}
    Let $k\ge 1$ and let $S\subseteq V$. There is an $O(n|S|^{\frac{1}{k}})$-space distance oracle that given two vertices $u,v\in V$, returns in $O(k)$ time an estimation $\hat{d}(u,v)$ such that $d(u,v)\le \hat{d}(u,v) \le 2\min(h_S(u),h_S(v))+(2k-1)d(u,v)$. The distance oracle is constructed in $O(m|S|^{\frac{1}{k}})$ time.
\end{lemma}

\begin{proof}
    The main idea behind our parameterized distance oracle is to construct the standard Thorup–Zwick hierarchy \cite{DBLP:journals/jacm/ThorupZ05}, but to skip the computation related to the set $A_1$($=S$), that is, we do not compute $B_0(u)$ for every $u\in V$. By omitting this step, we obtain a faster construction time, at the cost of increasing the stretch by an additional term of $2\min(h_S(u), h_S(v))$. Formally, the construction proceeds as follows.

    Let $A_0=V$, $A_1=S$, $A_{k+1}=\emptyset$. 
    For every $2\le i \le k$ the set $A_i\subseteq A_{i-1}$ is a random set such that $|A_i|=|A_{i-1}|\cdot |S|^{-i/k}$. 
    We let $B_0(u)=\{u\}$ for every $u\in V$.
    As presented in~\cite{DBLP:conf/icalp/RodittyTZ05}, we compute $C(u)$ for every $u\in S$ in $O(m|S|^{\frac{1}{k}})$ time in order to save for every $u\in V$ and $1 \le i \le k$ the sets $B_i(u)$.
    From~\Cref{L-TZ-Size}, we have that $\sum_{u\in V}{|B(u)|} \le O(n|S|^{\frac{1}{k}})$, thus, the distance oracle uses $O(n|S|^{\frac{1}{k}})$ space.

    The query of the distance oracle returns $\min(\ADOQuery(u,v), \ADOQuery(v,u))$. To prove its correctness, we first prove by induction for every $i\in [1,k+1]$ that either $\hat{d}(u,v) \le 2\min(h_S(u),h_S(v)) + (2i-3)d(u,v)$ or $\min(h_i(u),h_i(v)) \le \min(h_S(u),h_S(v)) + (i-1)\cdot d(u,v)$.
    
    For $i=1$, since $h_1(\cdot)=h_S(\cdot)$ we get that $\min(h_1(u),h_1(v)) = \min(h_S(u),h_S(v))$ and the claim holds.
    Next, we assume the claim holds for $i-1$ and prove the claim for $i$. 
    From the fact that the claim holds for $i-1$, we get that either $$\hat{d}(u,v) \le 2\min(h_S(u),h_S(v)) + (2i-5)d(u,v)$$ 
    or $$\min(h_{i-1}(u),h_{i-1}(v)) \le \min(h_S(u),h_S(v)) + (i-2)\cdot d(u,v)$$
    If $\hat{d}(u,v) \le 2\min(h_S(u),h_S(v)) + (2i-5)d(u,v)$ then $\hat{d}(u,v) \le 2\min(h_S(u),h_S(v)) + (2i-3)d(u,v)$ and the claim holds. Otherwise, we have that $$\min(h_{i-1}(u),h_{i-1}(v)) \le \min(h_S(u),h_S(v)) + (i-2)\cdot d(u,v)$$
    From~\Cref{L-ADO.Query-Correctness-2} it follows that either 
    $$\hat{d}(u,v)\le 2\min(h_{i-1}(u),h_{i-1}(v)) + d(u,v) \text{ or } \min(h_i(u),h_i(v)) \le d(u,v) + \min(h_{i-1}(u),h_{i-1}(v))$$
    Since $\min(h_{i-1}(u),h_{i-1}(v)) \le \min(h_S(u),h_S(v)) + (i-2)\cdot d(u,v)$ we get that either 
    $$\hat{d}(u,v) \le 2\min(h_S(u),h_S(v)) + (2i-3)\cdot d(u,v) \text{ or }$$
    $$\min(h_i(u),h_i(v)) \le \min(h_S(u),h_S(v)) + (i-2)\cdot d(u,v) + d(u,v) = \min(h_S(u),h_S(v)) + (i-1)\cdot d(u,v),$$ as required.
    
    We complete the proof since for $i=k+1$ we have that $\hat{d}(u,v) \le 2\min(h_S(u),h_S(v)) + (2k-1)\cdot d(u,v)$ since $\min(h_{k+1}(u),h_{k+1}(v)) \le \min(h_S(u),h_S(v)) + (i-1)d(u,v)$ does not hold since $A_{k+1}=\emptyset$ and therefore $h_{k+1}(v)=h_{k+1}(u)=\infty$.
\end{proof}

Next, we present the second parameterized  distance oracle 
that supports queries only between vertices from $S$.
Notice that this oracle uses much less space than the first oracle. 

\begin{lemma}\label{L-ADO-S-Only}
    Let $k\ge 1$ and let $S\subseteq V$. There is an $O(|S|^{1+\frac{1}{k}})$-space distance oracle that given two vertices $u,v\in S$, returns in $O(k)$ time an estimation $\hat{d}(u,v)$ such that $d(u,v)\le \hat{d}(u,v) \le (2k-1)d(u,v)$. The distance oracle is constructed in $O(m|S|^{\frac{1}{k}})$ time.
\end{lemma}
\begin{proof}
    The construction and the query are the same as in the distance oracle from~\Cref{L-ADO-G_S}. However, instead of saving $B_i(u)$, for every $u\in V$ and $0\leq i \leq k$, we save $B_i(s)$ only for every $s\in S$. Thus, the space is $O(|S|^{1+\frac{1}{k}})$.
    The correctness of the query follows from the fact that $\ADOQuery(u,v)$ only uses the bunches of $u$ and $v$. Thus, if $s_1,s_2\in S$ then we have the information needed in $\ADOQuery(s_1,s_2)$.
\end{proof}

Since the distance oracle of~\Cref{L-ADO-G_S} and the distance oracle of~\Cref{L-ADO-S-Only} differ only by the supported vertex pairs, we refer to the oracle of~\Cref{L-ADO-G_S} as $\ADO_P(G,k, S)$ and to the oracle of~\Cref{L-ADO-S-Only} as $\ADO_{P'}(G,k, S)$.

\subsection{A new hierarchy of distance oracles}\label{S-hado}
In this section, we present a new distance oracle composed of a hierarchy of $t=\ceil{\log\log{n}}$ distance oracles. The distance oracle computes a set $S_t$. The query of the distance oracle guarantees a $(2k-1)$-stretch for distances in $G_{S_t}$.

The distance oracle is constructed as follows. We get as an input a graph and two parameters $k$ and $x_0$, where  $k\geq 3$ and   $\frac{1}{k}\leq x_0<1$. 
First, we compute a random set  $S_0\subseteq V$  such that $|S_0|=n^{1-x_0}$ and construct $\ADO(G_{S_0},k)$.

Next, for every $1\leq i \leq t$, we compute a random set $S_i\subseteq S_{i-1}$ such that $|S_i|=n^{1-x_i}$, where $x_i=x_0-\frac{1}{k(k-1)}+\frac{x_{i-1}}{k-1}$, and construct $\ADO_P(G_{S_i},k-1,S_{i-1})$.

Our distance oracle stores the distance oracle $\ADO(G_{S_0},k)$
and the distance oracles $\ADO_P(G_{S_i},k-1, S_{i-1})$, for every $1\leq i \leq t$.
The construction algorithm is presented in \Cref{A-Construction-t}.
The query algorithm returns $$\min(\ADO(G_{S_0},k)\Query(u,v),\min_{i\in [t]}(\ADO_P(G_{S_i},k-1,S_{i-1})\Query(u,v)))$$
\begin{algorithm2e}[t] 
\caption{$\hADO.\Construct(G,k,x_0)$}\label{A-Construction-t}
$S_0 \gets \Sample(V,n^{1-x_0})$ \\
Construct and store $\ADO(G_{S_0},k) $\\
$t \gets \ceil{\log\log{n}}$ \\
\For{$i\in [1,t]$} {
    $x_i \gets x_0-\frac{1}{k(k-1)}+\frac{x_{i-1}}{k-1}$ \\
    $S_i \gets \Sample(S_{i-1}, n^{1-x_i})$ \\
    Construct and store $\ADO_P(G_{S_i}, k-1, S_{i-1})$\\
}
\Return $S_t$
\end{algorithm2e}

\colorlet{cset}{blue!60!black}        
\definecolor{cpruned}{RGB}{208,2,27}  
\colorlet{cin}{green!50!black}        
\colorlet{cbox}{gray!55}

\begin{figure}[t]
\centering
\begin{tikzpicture}[
    vert/.style   = {circle, draw=black, fill=white, minimum size=3.4pt, inner sep=0pt, line width=0.5pt},
    svert/.style  = {circle, draw=cset, fill=cset, minimum size=5pt, inner sep=0pt},
    bvert/.style  = {circle, draw=black, fill=black!12, minimum size=9pt, inner sep=0pt, line width=0.9pt, font=\tiny},
    gbox/.style   = {rounded corners=4pt, draw=cbox, line width=0.9pt, fill=gray!3},
    nedge/.style  = {gray!55, line width=0.6pt},
    nxedge/.style = {cpruned, dashed, line width=0.6pt}, 
    pedge/.style  = {black, line width=1.3pt, line cap=round},
    xedge/.style  = {cpruned, dashed, line width=1.1pt}  
]

\def\hverts{%
  \coordinate (u) at (0.20,0.90); \coordinate (a) at (0.64,0.90);
  \coordinate (b) at (1.08,0.90); \coordinate (c) at (1.52,0.90);
  \coordinate (d) at (1.96,0.90); \coordinate (v) at (2.40,0.90);
  \coordinate (e) at (0.45,1.55); \coordinate (f) at (1.30,1.68);
  \coordinate (g) at (2.15,1.55); \coordinate (h) at (0.75,0.22);
  \coordinate (i) at (1.50,0.16); \coordinate (j) at (2.20,0.28);%
}

\begin{scope}[xshift=0cm]
  \node[gbox, minimum width=2.85cm, minimum height=2.25cm] at (1.30,0.90) {};
  \hverts
  
  \draw[nedge] (e)--(u); \draw[nedge] (h)--(a); \draw[nedge] (j)--(v);
  \draw[nxedge] (f)--(b); \draw[nxedge] (g)--(d); \draw[nxedge] (i)--(c);
  \draw[nxedge] (f)--(g); \draw[nxedge] (h)--(c);
  \draw[nxedge] (e)--(b); \draw[nxedge] (g)--(j); \draw[nxedge] (g)--(v); \draw[nxedge] (e)--(f);
  
  \draw[pedge] (u)--(a); \draw[pedge] (a)--(b);
  \draw[xedge] (b)--(c); \draw[xedge] (c)--(d); \draw[xedge] (d)--(v);
  
  \foreach \p in {u,a,b,c,d,v,e,f,g,h,i,j} {\node[vert] at (\p) {};}
  \node[bvert] at (u) {$u$}; \node[bvert] at (v) {$v$};
  
  \node[font=\scriptsize, anchor=south] at (1.30,2.12) {$G_{S_0}$};
  \node[font=\scriptsize, gray!55!black] at (1.30,-0.42) {$|E_{S_0}|=n^{1+x_0}$};
  \node[font=\scriptsize, gray!55!black] at (1.30,-0.77) {(sparsest graph)};
  \node[font=\scriptsize, cpruned] at (1.30,-1.14) {$P(u,v)\not\subseteq G_{S_0}$};
\end{scope}

\begin{scope}[xshift=3.2cm]
  \node[gbox, minimum width=2.85cm, minimum height=2.25cm] at (1.30,0.90) {};
  \hverts
  
  \draw[nedge] (e)--(u); \draw[nedge] (h)--(a); \draw[nedge] (j)--(v);
  \draw[nedge] (f)--(b); \draw[nedge] (g)--(d);
  \draw[nxedge] (i)--(c); \draw[nxedge] (f)--(g); \draw[nxedge] (h)--(c);
  \draw[nxedge] (e)--(b); \draw[nxedge] (g)--(j); \draw[nxedge] (g)--(v); \draw[nxedge] (e)--(f);
  
  \draw[pedge] (u)--(a); \draw[pedge] (a)--(b); \draw[pedge] (b)--(c);
  \draw[xedge] (c)--(d); \draw[xedge] (d)--(v);
  
  \foreach \p in {u,a,b,c,d,v,e,f,g,h,i,j} {\node[vert] at (\p) {};}
  \foreach \p in {e,f,g,h,i} {\node[svert] at (\p) {};}
  \node[bvert] at (u) {$u$}; \node[bvert] at (v) {$v$};
  
  \node[font=\scriptsize, anchor=south] at (1.30,2.12) {$G_{S_1}$};
  \node[font=\scriptsize, cset] at (1.30,-0.42) {with sample $S_0$};
  \node[font=\scriptsize, gray!55!black] at (1.30,-0.77) {$|E_{S_1}|=n^{1+x_1}$};
  \node[font=\scriptsize, cpruned] at (1.30,-1.14) {$P(u,v)\not\subseteq G_{S_1}$};
\end{scope}

\begin{scope}[xshift=6.4cm]
  \node[gbox, minimum width=2.85cm, minimum height=2.25cm] at (1.30,0.90) {};
  \hverts
  
  \draw[nedge] (e)--(u); \draw[nedge] (f)--(b); \draw[nedge] (h)--(a); \draw[nedge] (j)--(v);
  \draw[nedge] (f)--(g); \draw[nedge] (g)--(d); \draw[nedge] (h)--(c); \draw[nedge] (i)--(c);
  \draw[nxedge] (e)--(b); \draw[nxedge] (g)--(j); \draw[nxedge] (g)--(v); \draw[nxedge] (e)--(f);
  
  \draw[pedge] (u)--(a)--(b)--(c); \draw[pedge] (c)--(d);
  \draw[xedge] (d)--(v);
  
  \foreach \p in {u,a,b,c,d,v,e,f,g,h,i,j} {\node[vert] at (\p) {};}
  \foreach \p in {f,g,h} {\node[svert] at (\p) {};}
  \node[bvert] at (u) {$u$}; \node[bvert] at (v) {$v$};
  
  \node[font=\scriptsize, anchor=south] at (1.30,2.12) {$G_{S_2}$};
  \node[font=\scriptsize, cset] at (1.30,-0.42) {with sample $S_1$};
  \node[font=\scriptsize, gray!55!black] at (1.30,-0.77) {$|E_{S_2}|=n^{1+x_2}$};
  \node[font=\scriptsize, cpruned] at (1.30,-1.14) {$P(u,v)\not\subseteq G_{S_2}$};
\end{scope}

\node[font=\large] at (9.80,0.90) {$\cdots$};

\begin{scope}[xshift=10.6cm]
  \node[gbox, draw=cin, fill=cin!7, minimum width=2.85cm, minimum height=2.25cm] at (1.30,0.90) {};
  \hverts
  
  \draw[nedge] (e)--(u); \draw[nedge] (f)--(b); \draw[nedge] (h)--(a); \draw[nedge] (j)--(v);
  \draw[nedge] (f)--(g); \draw[nedge] (g)--(d); \draw[nedge] (h)--(c); \draw[nedge] (i)--(c);
  \draw[nedge] (e)--(b); \draw[nedge] (g)--(j); \draw[nedge] (g)--(v); \draw[nedge] (e)--(f);
  
  \draw[pedge] (u)--(a)--(b)--(c)--(d)--(v);
  
  \foreach \p in {u,a,b,c,d,v,e,f,g,h,i,j} {\node[vert] at (\p) {};}
  \foreach \p in {g} {\node[svert] at (\p) {};}
  \node[bvert] at (u) {$u$}; \node[bvert] at (v) {$v$};
  
  \node[font=\scriptsize, anchor=south] at (1.30,2.12) {$G_{S_i}$};
  \node[font=\scriptsize, cin!55!black] at (1.30,1.86) {$P(u,v)\subseteq G_{S_i}$};
  \node[font=\scriptsize, cset] at (1.30,-0.42) {with sample $S_{i-1}$};
  \node[font=\scriptsize, gray!55!black] at (1.30,-0.77) {$|E_{S_i}|=n^{1+x_i}$};
\end{scope}

\node[font=\scriptsize, cin!55!black, align=center] at (11.90,-1.42)
  {($i$ = first such that $P(u,v)\subseteq G_{S_i}$)\\[1pt] $\Rightarrow\ \hat d(u,v)\le(2k-1)\,d(u,v)$};


\end{tikzpicture}
\caption{%
Illustration of the sampling hierarchy used by $\hADO$: Starts with $G_{S_0}$ - the sparsest graph in which we construct $\ADO(G_{S_0},k)$ without a sample, then in each iteration the sample $S_i$ is getting smaller and each subsequent $G_{S_{i}}$ is getting denser. In red are the missing edges from the graph in $G_{S_i}$. Level $i$ is the first level such that $P(u,v)\subseteq G_{S_i}$, and we have $\ADO_P(G_{S_i},k-1,S_{i-1})\Query(u,v)\le(2k-1)d(u,v)$ (\Cref{L-Construct-to-T-correction}).
}
\label{fig:hierarchy}
\end{figure}

In the next lemma, we analyze the running time of the construction algorithm.
\begin{lemma}\label{L-construct-time-up-to-T}
    The distance oracle is constructed in $O((m+kn^{1+x_0+\frac{1}{k}})\log\log{n})$ time.
\end{lemma}
\begin{proof}
    Using~\Cref{L-G_S-size} and the fact that $|S_0|=n^{1-x_0}$ we get that $|E_{S_0}|=O(n^{1+x_0})$. Thus, constructing $\ADO(G_{S_0},k)$ takes $O(kn^{1+x_0+\frac{1}{k}})$ time.
    
    For every $i\in [1,t]$, since $|S_i|=n^{1-x_i}$, it follows from~\Cref{L-G_S-size} that $|E_{S_i}|=O(n^{1+x_i})$ and constructing $\ADO_P(G_{S_i}, k-1, S_{i-1})$ takes  $O(k|E_{S_i}|\cdot |S_{i-1}|^{1/(k-1)})=O(kn^{1+x_i+\frac{1-x_{i-1}}{k-1}})$ time.
    By definition $x_i=x_0-\frac{1}{k(k-1)}+x_{i-1}/(k-1)$, thus, we get that the construction time is
    $O(kn^{1+x_0-\frac{1}{k(k-1)}+\frac{x_{i-1}}{k-1}+\frac{1-x_{i-1}}{k-1}})=O(kn^{1+x_0+\frac{1}{k}})$ for every $i\in [t]$. In addition, computing $G_{S_i}$ takes $O(m)$ time, for every $i\in [1,t]$.
    Since $t=\ceil{\log\log{n}}$ we get that the running time is $O((m+kn^{1+x_0+\frac{1}{k}})\log\log{n})$, as required.
\end{proof}

Next, we analyze the space used by the distance oracle. 

\begin{lemma}\label{T-Construct-to-T-space}
    The distance oracle uses $O(n^{1+\frac{1}{k}})$ space.
\end{lemma}
\begin{proof}
    From~\Cref{L-Regular-ADO-Construction} we have that  $\ADO(G_{S_0},k)$ uses $O(n^{1+\frac{1}{k}})$ space. From~\Cref{L-ADO-G_S} we have that for every $i\in [1,t]$, $\ADO_P(G_{S_i},k-1,S_{i-1})$ uses $O(\sum_{i\in[t]}{n|S_{i-1}|^\frac{1}{k-1}})$ space. Since $x_i > x_0 \ge \frac{1}{k}$, we have that $O(\sum_{i\in[t]}{n|S_i|^\frac{1}{k-1}}) = O(n^{1+\frac{1}{k}})$, as required.
\end{proof}
We now turn to prove that the stretch of $\hADO$ for distances in $G_{S_t}$ is $2k-1$. 
\begin{lemma}\label{L-Construct-to-T-correction}
    If $P(u,v) \subseteq G_{S_t}$ then $\hat{d}(u,v) \le (2k-1)d(u,v)$
\end{lemma}
\begin{proof}
    We divide the proof into two cases. The case that $P(u,v)\subseteq G_{S_0}$ and the case that $P(u,v)\not\subseteq G_{S_0}$.
    Consider the case that $P(u,v)\subseteq G_{S_0}$ from the correctness of the query of $\ADO(G_{S_0},k)$ we have that $\hat{d}(u,v) \le (2k-1)d(u,v)$, as required.
    Consider now the case that $P(u,v)\not\subseteq G_{S_0}$. From~\Cref{L-No-Intersect-Weighted} it follows that $\max(h_{S_0}(u),h_{S_0}(v))\le d(u,v)$. 
    
    Next, we show by induction that for every $i\in [1,t]$, either 
    $\max(h_{S_i}(u),h_{S_i}(v)) \le d(u,v)$ or that $\hat{d}(u,v) \le (2k-1)d(u,v)$ as required.
    Assume that $\max(h_{S_{i-1}}(u),h_{S_{i-1}}(v)) \le d(u,v)$, we show that either
    $\max(h_{S_i}(u),h_{S_i}(v)) \le d(u,v)$ or that $\hat{d}(u,v) \le (2k-1)d(u,v)$.
    We divide the proof into two cases. The case that $P(u,v) \subseteq G_{S_i}$ and the case that $P(u,v) \not\subseteq G_{S_i}$.
    Consider the case that $P(u,v) \subseteq G_{S_i}$.
    From~\Cref{L-ADO-G_S} we have that $\ADO_P(G_{S_i},k-1,S_{i-1})\Query(u,v) \le 2\min(h_{S_{i-1}}(u), h_{S_{i-1}}(v)) + (2(k-1)-1)d(u,v)$. From the induction assumption we know that $\max(h_{S_{i-1}}(u),h_{S_{i-1}}(v)) \le d(u,v)$, thus we get that $\hat{d}(u,v) \le 2\min(h_{S_{i-1}}(u),h_{S_{i-1}}(v)) + (2(k-1)-1)d(u,v) \le (2k-1)d(u,v)$, as required.
    Consider now the case that $P(u,v) \not\subseteq G_{S_i}$.
    From~\Cref{L-No-Intersect-Weighted} it follows that $\max(h_{S_i}(u),h_{S_i}(v))\le d(u,v)$, as required.
\end{proof}

The size of $S_t$ is $n^{1-x_t}$. Therefore, to bound $|S_t|$ we analyze the series $x_i=x_0-\frac{1}{k(k-1)}+\frac{x_{i-1}}{k-1}$.

\begin{cclaim}\label{C-x_i-serie}
     If $x_i=x_0-\frac{1}{k(k-1)}+\frac{x_{i-1}}{k-1}$ for $i\ge 1$, then  
     $x_j=\frac{k-1}{k-2}x_0-\frac{1}{k(k-2)}+\frac{1-x_0k}{k(k-2)}\cdot (\frac{1}{k-1})^j$ for every  $j\ge0$.
\end{cclaim}
\begin{proof}
    We prove the claim by induction. For $j=0$, we have that 
    \begin{align*}
        \frac{k-1}{k-2}x_0-\frac{1}{k(k-2)}+\left(\frac{1}{k-1}\right)^0\frac{1-x_0k}{k(k-2)}&=\frac{k-1}{k-2}x_0-\frac{1}{k(k-2)}+\frac{1}{k(k-2)}-\frac{x_0k}{k(k-2)} 
        \\&=\frac{k-1}{k-2}x_0-\frac{1}{k-2}x_0=x_0,
    \end{align*}
    as required.

    Next, we assume that the claim holds for $j-1$ and prove it for $j$. Thus, $x_{j-1}=\frac{k-1}{k-2}x_0-\frac{1}{k(k-2)}+(\frac{1}{k-1})^{j-1}\frac{(1-x_0k)}{k(k-2)}$.
    From the recursive definition of $x_j$ we have that $x_j=x_0-\frac{1}{k(k-1)}+\frac{x_{j-1}}{k-1}$. 
    Thus, we get that:
    \begin{align*}
        x_j&=x_0-\frac{1}{k(k-1)}+\frac{x_{j-1}}{k-1}=x_0-\frac{1}{k(k-1)}+\frac{\frac{k-1}{k-2}x_0-\frac{1}{k(k-2)}+\frac{(1-x_0k)}{k(k-2)}\cdot (\frac{1}{k-1})^{j-1}}{k-1}\\
        &=x_0-\frac{1}{k(k-1)}+\frac{x_0}{k-2}-\frac{1}{k(k-1)(k-2)}+\frac{(1-x_0k)}{k(k-2)}\cdot (\frac{1}{k-1})^{j}\\
        &=\frac{k-1}{k-2}x_0-\frac{1+k-2}{k(k-1)(k-2)}+\frac{(1-x_0k)}{k(k-2)}\cdot (\frac{1}{k-1})^{j}=\frac{k-1}{k-2}x_0-\frac{1}{k(k-2)}+\frac{(1-x_0k)}{k(k-2)}\cdot(\frac{1}{k-1})^{j},
    \end{align*}
    as required.

\end{proof}
Next, we move to bound $x_t$ and $|S_t|$ in the following property, by applying $t=\ceil{\log\log{n}}$ in \Cref{C-x_i-serie}.
\begin{property}\label{P-Limit-of-x-to-t}
    $x_{t}=\frac{k-1}{k-2}x_0-\frac{1}{k(k-2)}-O(\frac{1}{\log{n}})$ and $|S_t| = O(n^{1-\frac{k-1}{k-2}x_0+\frac{1}{k(k-2)}})$
\end{property}
\begin{proof}
    By applying $t=\ceil{\log\log{n}}$ in \Cref{C-x_i-serie} we get that 
    \begin{align*}
        x_t &= \frac{k-1}{k-2}x_0-\frac{1}{k(k-2)}+\frac{1-x_0k}{k(k-2)}\cdot (\frac{1}{k-1})^t \ge \frac{k-1}{k-2}x_0-\frac{1}{k(k-2)}+\frac{1-x_0k}{k(k-2)}\cdot (\frac{1}{k-1})^{\log\log{n}} \\
        &\ge \frac{k-1}{k-2}x_0-\frac{1}{k(k-2)} - O(1)\cdot O(1/2^{\log\log{n}})=\frac{k-1}{k-2}x_0-\frac{1}{k(k-2)} - O(1/\log{n}),
    \end{align*}
    as required.
    Since $|S_t|=O(n^{1-x_t})$, we get that $|S_t|=O(n^{\frac{k-1}{k-2}x_0-\frac{1}{k(k-2)} + O(1/\log{n})})=O(n^{\frac{k-1}{k-2}x_0-\frac{1}{k(k-2)}})$, as required.
\end{proof}
In the next theorem, we summarize the properties of the hierarchical distance oracle presented in this section.

\begin{theorem}
    \label{T-Construct-to-T}
    Let $\frac{1}{k}\le x_0<1$ be a parameter, let $k>2$ and let $t=\ceil{\log\log{n}}$. Let $S_t\subseteq V$ be a random set with $O(n^{\frac{k-1}{k-2}x_0+\frac{1}{k(k-2)}})$ vertices.
    There is an algorithm that constructs in $O((m+kn^{1+x_0+\frac{1}{k}})\log\log{n})$-time a distance oracle with $O(n^{1+\frac{1}{k}})$ space, such that for every $u,v\in V$ such that $P(u,v)\subseteq G_{S_t}$ it returns $\hat{d}(u,v) \le (2k-1)d(u,v)$ in $O(k\log\log{n})$-time.
\end{theorem}
\Cref{T-Construct-to-T} follows from~\Cref{L-construct-time-up-to-T},~\Cref{T-Construct-to-T-space},~\Cref{L-Construct-to-T-correction},  and~\Cref{P-Limit-of-x-to-t}.

Throughout the paper, we refer to the hierarchical distance oracle of~\Cref{T-Construct-to-T} as $\hADO(G,k,x_0)$.

\section{Faster constructions for weighted graphs}\label{S-Construction-weighted}
In this section, we present our main results, faster constructions for the classic $(2k-1)$-stretch distance oracle with $O(n^{1+\frac{1}{k}})$ space. 
In \Cref{S-Construction-Sparse} we present our simplest construction, a subquadratic time construction algorithm for graphs with $m=O(n^{2-\eps})$ for every $\eps>0$ and $k\geq 3$.

Next, in \Cref{S-construction-spanner} and \Cref{S-construction-spanner-ado}, we present two construction algorithms with a running time of the form $O(m+n^{1+f(k)})$. The first algorithm is faster for $4\leq k\leq 15$ and the second for $k\geq 16$.

\subsection{Truly subquadratic construction time for every $k>2$}\label{S-Construction-Sparse}
In this section, we present the first truly subquadratic construction time for every $2<k<6$. 
Specifically, we prove:
\Reminder{T-Construction-no-spanner}

The distance oracle is constructed as follows.
Let $x_0$ be a parameter to be determined later, where $\frac{1}{k} \le x_0 < 1$.
We start by constructing $\hADO(G,k,x_0)$ and obtain the set $S_t$. We then 
construct $\ADO_P(G,k-1,S_t)$.
We return the two constructed distance oracles.
The query algorithm returns $\min(\hADO(G,k,x_0)\Query(u,v),\ADO_P(G,k-1,S_t)\Query(u,v))$.

Next, we analyze the construction time and the space usage of the distance oracle.
\begin{lemma}\label{L-Construction-time-no-spanner}
    The distance oracle is constructed in $O(\max(n^{1+2/k}, m^{1-\frac{1}{k-1}}n^{\frac{2}{k-1}})\log\log{n})$-time and uses $O(n^{1+\frac{1}{k}})$-space.
\end{lemma}
\begin{proof}
    From~\Cref{T-Construct-to-T} it follows that constructing $\hADO(G,k,x_0)$ takes $O((m+n^{1+x_0+\frac{1}{k}})\log\log{n})$ time.
    From Property~\Cref{P-Limit-of-x-to-t} it follows that $|S_t| = O(n^{1-\frac{k-1}{k-2}x_0+\frac{1}{k(k-2)}})$. From~\Cref{L-ADO-G_S} it follows that constructing $\ADO_P(G,k-1,S_t)$ takes $O(m|S_t|^{1/(k-1)})=O(mn^{\frac{1-\frac{k-1}{k-2}x_0+\frac{1}{k(k-2)}}{k-1}})$ time.
    
    Thus, the total construction time is $O((m+n^{1+x_0+\frac{1}{k}})\log\log{n}+mn^{\frac{1-\frac{k-1}{k-2}x_0+\frac{1}{k(k-2)}}{k-1}})$.
    To minimize the construction time, we choose $x_0$ such that 
    $n^{1+x_0+\frac{1}{k}} = mn^{\frac{1-\frac{k-1}{k-2}x_0+\frac{1}{k(k-2)}}{k-1}}$.
    From 
    \href{https://www.wolframalpha.com/input?i=1%2Bx%2B1%2Fk+%3D+a%2B%281-%28k-1%29x%2F%28k-2%29%2B1%2F%28k%28k-2%29%29%29%2F%28k-1%29}{\footnotemark}\footnotetext{https://www.wolframalpha.com/input?i=1\%2Bx\%2B1\%2Fk+\%3D+a\%2B\%281-\%28k-1\%29x\%2F\%28k-2\%29\%2B1\%2F\%28k\%28k-2\%29\%29\%29\%2F\%28k-1\%29, $x:=x_0$, $a:=\log_n(m), k:=k$}
    we get that for $x_0=\frac{\log_n{m}k^2-2\log_n{m}k-k^2+2k+1}{k(k-1)}$ the terms are equal.
    $\hADO$ requires that $x_0 \ge \frac{1}{k}$, thus we let $x_0=\max(\frac{1}{k}, \frac{\log_n{m}k^2-2\log_n{m}k-k^2+2k+1}{k(k-1)})$.
    
    Increasing $x_0$ increases $n^{1+x_0+\frac{1}{k}}$ while decreasing $mn^{\frac{1-\frac{k-1}{k-2}x_0+\frac{1}{k(k-2)}}{k-1}}$, therefore we get that
    $O(n^{1+\max(\frac{1}{k}, \frac{\log_n{m}k^2-2\log_n{m}k-k^2+2k+1}{k(k-1)})+\frac{1}{k}}) \ge O(mn^{\frac{1-\frac{k-1}{k-2}x_0+\frac{1}{k(k-2)}}{k-1}})$, and thus the construction time is:
    \begin{align*}
        &O((m+n^{1+x_0+\frac{1}{k}})\log\log{n}+mn^{\frac{1-\frac{k-1}{k-2}x_0+\frac{1}{k(k-2)}}{k-1}}) 
        = O(n^{1+x_0+\frac{1}{k}}\log\log{n})
        \\&=O(n^{1+\max(\frac{1}{k}, \frac{\log_n{m}k^2-2\log_n{m}k-k^2+2k+1}{k(k-1)})+\frac{1}{k}}\log\log{n})\\
        &=O(\max(n^{1+2/k}, n^{1+\frac{1}{k}+\frac{\log_n{m}k^2-2\log_n{m}k-k^2+2k+1}{k(k-1)}})\log\log{n}) \\
        &= O(\max(n^{1+2/k}, n^{\log_n{m}\frac{k^2-2k}{k(k-1)}+1+\frac{1}{k}+\frac{-k^2+2k+1}{k(k-1)}}) \log\log{n}) \\
        &= O(\max(n^{1+2/k}, n^{\log_n{m}(1-\frac{1}{k-1})+\frac{(k^2-k)+(k-1)-k^2+2k+1}{k(k-1)}})\log\log{n})
        \\&= O(\max(n^{1+2/k}, m^{1-\frac{1}{k-1}}n^{\frac{2}{k-1}})\log\log{n}),
    \end{align*}
    as required.
    
    From~\Cref{T-Construct-to-T} and the fact that $x_0 \ge \frac{1}{k}$ it follows that $\hADO(G,k,x_0)$ uses $O(n^{1+\frac{1}{k}})$ space. 
    From~\Cref{L-ADO-G_S} it follows that $\ADO_P(G,k-1,S_t)$ uses $O(n|S_t|^{\frac{1}{k-1}})=O(n^{1+\frac{1-x_t}{k-1}})$ space. Since $x_t > x_0 \ge \frac{1}{k}$, we have that $O(n^{1+\frac{1-x_t}{k-1}}) = O(n^{1+\frac{1}{k}})$, as required.
\end{proof}
Next, we prove that the stretch is $2k-1$.
\begin{lemma}\label{L-Correctness-no-spanner}
    $\hat{d}(u,v) \le (2k-1)d(u,v)$
\end{lemma}
\begin{proof}
    We divide the proof into two cases. The case that $P(u,v)\subseteq G_{S_t}$ and the case that $P(u,v)\not\subseteq G_{S_t}$.
    Consider the case that $P(u,v)\subseteq G_{S_t}$. From~\Cref{T-Construct-to-T} we have that $\hat{d}(u,v) \le \hADO(G_{S_t},k, x_0)\Query(u,v) \le (2k-1)d(u,v)$, as required.
    
    Consider now the case that $P(u,v)\not\subseteq G_{S_t}$. 
    From~\Cref{L-No-Intersect-Weighted} it follows that $\max(h_{S_t}(u),h_{S_t}(v))\le d(u,v)$. 
    From~\Cref{L-ADO-G_S} 
    it follows that $\ADO_P(G,k-1,S_t)\Query(u,v) \le 2\min(h_{S_t}(u),h_{S_t}(v)) + (2(k-1)-1)d(u,v)$. Since $\max(h_{S_t}(u),h_{S_t}(v))\le d(u,v)$ we have that 
    \begin{align*}
        \hat{d}(u,v) \le \ADO_P(G,k-1,S_t)\Query(u,v) \le 2d(u,v) + (2(k-1)-1)d(u,v) \le (2k-1)d(u,v),
    \end{align*}
    as required.
\end{proof}
\Cref{T-Construction-no-spanner} follows from~\Cref{L-Construction-time-no-spanner} and~\Cref{L-Correctness-no-spanner}.
\subsection{$km+kn^{\frac{3}{2}+\frac{3}{4k-6}}$ construction time}\label{S-construction-spanner}
In this section, we present a faster construction algorithm for $k\ge 4$, which improves upon the results of~\cite{DBLP:conf/soda/Wulff-Nilsen12}.
We prove:
\Reminder{T-Construction-spanner}

First, we describe the construction of the distance oracle.
Similarly to the construction of \Cref{S-Construction-Sparse}, let $\frac{1}{k}\le x_0<1$ be a parameter to be determined later.
We start by constructing $\hADO(G,k,x_0)$, and save the returned set $S_t$.

Let $k'=k-2$. Using~\Cref{L-Spanner-2k-1} we create a $(2k'-1)$-stretch spanner $H$. 
For every $u\in V$, we add the edge $(u,p_{S_t}(u))$ to $H$ with weight $h_{S_t}(u)$.
For every $s_1,s_2\in S_t$ we compute and store the value of $d_H(s_1,s_2)$. The construction algorithm is presented in \Cref{A-Construction-Weighted-Small-k}.

Given $u,v\in V$, the query algorithm returns $\min(\hADO(G,k,x_0)\Query(u,v), h_{S_t}(u)+d_H(p_{S_t}(u),p_{S_t}(v))+h_{S_t}(v))$. 
Next, in the following two lemmas, we analyze the construction  time and the space usage of the distance oracle.

\begin{algorithm2e}[t]
\caption{$\ConstructADO(G,k)$}\label{A-Construction-Weighted-Small-k}
Construct $\hADO(G,k,x_0) $ \\
$H \gets \Spanner(G, k')$ [\Cref{L-Spanner-2k-1}]\\
$H \gets H \cup \{(u,p_{S_t}(u), h_{S_t}(u)) \mid u\in V\}$ \\
$d_H \gets \{s_1,s_2:d_H(s_1,s_2) \mid s_1,s_2\in S_t\}$
\end{algorithm2e}

\begin{lemma}\label{L-Construction-spanner}
    The distance oracle is constructed in 
    $O(k'm+(m+kn^{1+x_0+\frac{1}{k}})\log\log{n}+n^{2+\frac{1}{k'}-\frac{k-1}{k-2}x_0+\frac{1}{k(k-2)}})$-time.
\end{lemma}
\begin{proof}
    From~\Cref{T-Construct-to-T}, constructing $\hADO(G,k,x_0)$ takes $O((m+n^{1+x_0+\frac{1}{k}})\log\log{n})$ time. 
    From~\Cref{L-Spanner-2k-1} constructing $H$ takes $O(k'm)$ time.
    Computing $d_H(s_1,s_2)$ for every $s_1,s_2\in S_t$ takes $O(|E(H)|\cdot |S_t|)$.
    From Property~\ref{P-Limit-of-x-to-t} we have that $|S_t|=O(n^{1-\frac{k-1}{k-2}x_0+\frac{1}{k(k-2)}})$ and since $|E(H)|=O(n^{1+\frac{1}{k'}})$ we get that $O(|E(H)|\cdot |S_t|)=O(n^{1+\frac{1}{k'}+1-\frac{k-1}{k-2}x_0+\frac{1}{k(k-2)}})=O(n^{2+\frac{1}{k'}-\frac{k-1}{k-2}x_0+\frac{1}{k(k-2)}})$ time.
    Thus, the total runtime is $O(k'm+(m+kn^{1+x_0+\frac{1}{k}})\log\log{n}+n^{2+\frac{1}{k'}-\frac{k-1}{k-2}x_0+\frac{1}{k(k-2)}})$, as required.
\end{proof}

\begin{lemma}\label{L-ADO-Spanner-Space}
    If $x_0 \ge 1/2 - \frac{1}{k}+\frac{1}{k(k-1)}$ then the distance oracle uses $O(kn^{1+\frac{1}{k}})$-space. 
\end{lemma}
\begin{proof}
    From~\Cref{T-Construct-to-T} saving $\hADO(G,k,x_0)$  takes $O(kn^{1+\frac{1}{k}})$ space.
    Saving $d(s_1,s_2)$ for every $s_1,s_2\in S_t$ takes $O(|S_t|^2)$. 
    From Property~\ref{P-Limit-of-x-to-t} we know that 
    $|S_t|=O(n^{1-\frac{k-1}{k-2}x_0+\frac{1}{k(k-2)}})$. 
    Since $x_0 \ge 1/2 - \frac{1}{k}+\frac{1}{k(k-1)}$ it follows that $|S_t| \le O(n^{1-1/2+\frac{1}{2k}})$.
    Since $|S_t| \le O(n^{1-1/2+\frac{1}{2k}})$ we have that $O(|S_t|^2)\le O(n^{2(1-1/2+\frac{1}{2k})})=O(n^{1+\frac{1}{k}})$, as required.
\end{proof}
Next, we show that the stretch of the distance oracle is $2k-1$.
\begin{lemma}\label{L-ADO-Spanner-correctness}
    $\hat{d}(u,v) \le (2k-1)d(u,v)$
\end{lemma}
\begin{proof}
    Let $P(u,v)$ be the shortest path between $u$ and $v$.
    We divide the proof into two cases. 
    Consider the case that $P(u,v)\subseteq G_{S_t}$. From~\Cref{T-Construct-to-T} we know that $\hADO\Query(u,v) \le (2k-1)d(u,v)$, as required.
    Consider now the case that $P(u,v)\not\subseteq G_{S_t}$. From~\Cref{L-No-Intersect-Weighted} it follows that $\max(h_{S_t}(u),h_{S_t}(v)) \le d(u,v)$.
    From the triangle inequality, it follows that $d_H(p(u),p(v)) \le h(u) + d_H(u,v) + h(v)$.
    Since $H$ is a $(2k'-1)$-spanner we get that $d_H(u,v)\le (2k'-1)d(u,v)$. Thus, $d_H(p(u),p(v)) \le h(u) + h(v) + (2k'-1)d(u,v)$. Since $h(u)\le d(u,v)$ and $h(v) \le d(u,v)$ we get that $d_H(p(u),p(v)) \le h(u) + h(v) + (2k'-1)d(u,v) \le (2k'+1)d(u,v)$.
    
    Since $\hat{d}(u,v) \le h(u)+d_H(p(u),p(v))+h(v)$, we get that $\hat{d}(u,v) \le (2k'+3)d(u,v)$. Since $k'= k-2$ we get that $\hat{d}(u,v) \le (2k-1)d(u,v)$, as required.
\end{proof}

Next, we move to choose the value of $x_0$ that minimizes the construction time of our distance oracle.
From~\Cref{L-Construction-spanner} we know that the construction time is $O(k'm+kn^{1+x_0+\frac{1}{k}}+n^{2+\frac{1}{k'}-\frac{k-1}{k-2}x_0+\frac{1}{k(k-2)}})$. To minimize the running time we want to have that $1+x_0+\frac{1}{k} = 2+\frac{1}{k'}-\frac{k-1}{k-2}x_0+\frac{1}{k(k-2)}$. From \href{https://www.wolframalpha.com/input?i=y%3D%5Cfrac%7Bk-1%7D%7Bk-2%7Dx-%5Cfrac%7B1%7D%7Bk%28k-2%29%7D%2C+1%2Bx%2B1%2Fk%3D2%2B1%2F%28k-2%29-y}{\footnotemark}
\footnotetext{https://www.wolframalpha.com/input?i=y\%3D\%5Cfrac\%7Bk-1\%7D\%7Bk-2\%7Dx-\%5Cfrac\%7B1\%7D\%7Bk\%28k-2\%29\%7D\%2C+1\%2Bx\%2B1\%2Fk\%3D2\%2B1\%2F\%28k-2\%29-y, where $x=x_0$ and $y=x_t$}, we get that for $x_0=\frac{k^2-2k+3}{k(2k-3)}$, the terms are equal. We get a construction time of 
\begin{align*}
     O(k'm+(m+kn^{1+x_0+\frac{1}{k}})\log\log{n}+n^{2+\frac{1}{k'}-\frac{k-1}{k-2}x_0+\frac{1}{k(k-2)}}) 
    &= O((km+kn^{1+\frac{1}{k}+\frac{k^2-2k+3}{k(2k-3)}})\log\log{n}) \\
    &= O((km+kn^{1+\frac{1}{k}+(1/2+\frac{6-k}{2k(2k-3)})})\log\log{n}) \\
    &= O((km+kn^{\frac{3}{2}+\frac{3}{4k-6}})\log\log{n}),
\end{align*}
as required.
\Cref{T-Construction-spanner} follows from~\Cref{L-Construction-spanner},~\Cref{L-ADO-Spanner-Space} and~\Cref{L-ADO-Spanner-correctness}.
\subsection{$\sqrt{k}m+n^{1+\frac{2\sqrt{2}}{\sqrt{k}}}$ construction time} \label{S-construction-spanner-ado}

In this section, we present a faster construction algorithm for $k\ge 16$, which improves upon the results of~\cite{DBLP:conf/soda/Wulff-Nilsen12}.
We prove:
\Reminder{T-construction-spanner-ado}

First, we describe the construction of the distance oracle.
Similarly to the construction in \Cref{S-Construction-Sparse}, let $\frac{1}{k}\le x_0<1$ be a parameter to be determined later.
We start by constructing $\hADO(G,k,x_0)$, and save the returned set $S_t$.
Let $k'$ and $k''$ be two integer parameters to be determined later. Using~\Cref{L-Spanner-2k-1} we create a $(2k'-1)$-stretch spanner $H$, which we  augment with  an edge $(u,p_{S_t}(u))$ of weight $h_{S_t}(u)$, for every $u\in V$.
We construct $\ADO_{P'}(H,k'', S_t)$.
The construction algorithm is presented in \Cref{A-Construction-Weighted-Large-k}.

Given $u,v\in V$, the query algorithm returns: $$\min(\hADO(G,k,x_0)\Query(u,v), h_{S_t}(u)+\ADO_{P'}(H,k'',S_t)\Query(p_{S_t}(u),p_{S_t}(v))+h_{S_t}(v))$$
Next, in the following two lemmas, we analyze the construction  time and the space usage of the distance oracle.
\begin{algorithm2e}[t] 
\caption{$\ConstructADO(G,k)$}\label{A-Construction-Weighted-Large-k}
Construct $\hADO(G,k, x_0)$ \\
$H \gets \Spanner(G, k')$ [\Cref{L-Spanner-2k-1}]\\
$H \gets H \cup \{(u,p_{S_t}(u), h_{S_t}(u)) \mid u\in V\}$ \\
Construct $\ADO_{P'}(H,k'',S_t)$
\end{algorithm2e}

\begin{lemma}\label{L-Construction-spanner-Almost-Linear}
    The distance oracle is constructed in 
    $\Ot(k'm+k''n^{1+\frac{1}{k'}+\frac{1-\frac{k-1}{k-2}x_0+\frac{1}{k(k-2)}}{k''}}+kn^{1+x_0+\frac{1}{k}})$-time.
\end{lemma}
\begin{proof}
    From~\Cref{T-Construct-to-T}, we know that constructing $\hADO(G,k,x_0)$ takes $\Ot(m+n^{1+x_0+\frac{1}{k}})$ time.
    From~\Cref{L-Spanner-2k-1}, constructing $H$ takes $O(k'm)$ time.
    From~\Cref{L-ADO-S-Only}, constructing $\ADO_{P'}(H,k'', S_t)$ takes $O(k''|E(H)|\cdot |S_t|^{\frac{1}{k''}})$. From Property~\ref{P-Limit-of-x-to-t} we have that $|S_t|=O(n^{1-\frac{k-1}{k-2}x_0+\frac{1}{k(k-2)}})$. Thus, constructing $\ADO_{P'}$ takes $O(k''|E(H)|\cdot |S_t|^{\frac{1}{k''}}) = 
    O(k''n^{1+\frac{1}{k'}}\cdot n^{(\frac{1}{k''})(1-\frac{k-1}{k-2}x_0+\frac{1}{k(k-2)})}) = O(k''n^{1+\frac{1}{k'}+\frac{1-\frac{k-1}{k-2}x_0+\frac{1}{k(k-2)}}{k''}})$ time.
\end{proof}

\begin{lemma}\label{L-Almost-Linear-Space}
    If $(1-\frac{k-1}{k-2}x_0+\frac{1}{k(k-2)})(1+\frac{1}{k''}) \le 1+\frac{1}{k}$ then the distance oracle uses $O(kn^{1+\frac{1}{k}})$ space.
\end{lemma}
\begin{proof}
    From~\Cref{T-Construct-to-T} it follows that $\hADO(G,k,x_0)$  uses $O(kn^{1+\frac{1}{k}})$ space and
    from~\Cref{L-ADO-S-Only} it follows that $\ADO_{P'}(H,k'',S_t)$ uses $O(|S_t|^{1+\frac{1}{k''}})$ space.
    From Property~\ref{P-Limit-of-x-to-t} we know that $|S_t|=O(n^{1-\frac{k-1}{k-2}x_0+\frac{1}{k(k-2)}})$. Thus,  $\ADO_{P'}(H,k'',S_t)$ uses $O(n^{(1-\frac{k-1}{k-2}x_0+\frac{1}{k(k-2)})(1+\frac{1}{k''})})$ space. Since $(1-\frac{k-1}{k-2}x_0+\frac{1}{k(k-2)})(1+\frac{1}{k''}) \le 1+\frac{1}{k}$ we get that saving $\ADO_{P'}(H,k'',S_t)$ takes $O(n^{1+\frac{1}{k}})$, as required.
    
\end{proof}

Next, we show that the stretch of the distance oracle is $2k-1$.
\begin{lemma}\label{L-Almost-Linear-Correctness}
    If $2 + (2k''-1)(2k'+1) \le 2k-1$ then $\hat{d}(u,v) \le (2k-1)d(u,v)$
\end{lemma}
\begin{proof}
    Let $P(u,v)$ be the shortest path between $u$ and $v$.
    We divide the proof into two cases. 
    Consider the case that $P(u,v)\subseteq G_{S_t}$. From~\Cref{T-Construct-to-T} we know that $\hADO\Query(u,v) \le (2k-1)d(u,v)$, as required.
    
    Consider now the case that $P(u,v)\not\subseteq G_{S_t}$. From~\Cref{L-No-Intersect-Weighted} it follows that $\max(h_{S_t}(u),h_{S_t}(v)) \le d(u,v)$.
    From~\Cref{L-ADO-S-Only} we know that $\ADO_{P'}(H,k'', S_t)\Query(p(u),p(v)) \le (2k''-1)d_H(p(u),p(v))$.
    From the triangle inequality, and since $H$ is a $(2k'-1)$-spanner it follows that
    $d_H(p(u),p(v)) \le h(u)+(2k'-1)d(u,v)+h(v)\le (2k'+1)d(u,v)$. Where the last inequality holds since $\max(h(u),h(v)) \le d(u,v)$.

    Since $\hat{d}(u,v) \le h(u) + \ADO_{P'}(H,k'', S_t)\Query(p(u),p(v)) + h(v)$ and $\max(h(u),h(v)) \le d(u,v)$ we get that $\hat{d}(u,v) \le (2+(2k''-1)(2k'+1))d(u,v)$.
    Since $2 + (2k''-1)(2k'+1) \le 2k-1$ we get that $\hat{d}(u,v) \le (2k-1)d(u,v)$, as required.
\end{proof}
Next, we move to choose the value of the parameters $x_0,k',k''$ that minimizes the running time of our algorithm while maintaining the correctness of the query.

From~\Cref{L-Construction-spanner-Almost-Linear} we have that the construction time is $O(k'm+k''n^{1+\frac{1}{k'}+\frac{1-\frac{k-1}{k-2}x_0+\frac{1}{k(k-2)}}{k''}}+kn^{1+x_0+\frac{1}{k}})$. In addition, for the correctness,~\Cref{L-Almost-Linear-Correctness} requires that 
$2 + (2k''-1)(2k'+1) \le 2k-1$.
To minimize the term $n^{1+\frac{1}{k'}+\frac{1-\frac{k-1}{k-2}x_0+\frac{1}{k(k-2)}}{k''}}$ regarding $k'$ and $k''$ we want to have that $\frac{1}{k'}=\frac{1-\frac{k-1}{k-2}x_0+\frac{1}{k(k-2)}}{k''}$.
Thus, to find the optimal values for $x_0,k',k''$ we need to solve the following equations:
\begin{itemize}
    \item $1+\frac{1}{k'}+\frac{1-\frac{k-1}{k-2}x_0+\frac{1}{k(k-2)}}{k''} = 1+x_0+\frac{1}{k}$
    \item $2 + (2k''-1)(2k'+1) = 2k-1$
    \item $\frac{1}{k'}=\frac{1-\frac{k-1}{k-2}x_0+\frac{1}{k(k-2)}}{k''}$
\end{itemize}
From \href{https://www.wolframalpha.com/input?i=y%3D%5Cfrac%7Bk-1%7D%7Bk-2%7Dx-%5Cfrac%7B1%7D%7Bk%28k-2%29%7D%2C+1%2Bx%2B1%2Fk%3D1%2B1%2Fz%2B%281-y%29%2Fw%2C+2%2B%282w-1%29%282z%2B1%29+%3D+2k-1%2C+1%2Fz%3D%281-y%29%2Fw}{\footnotemark}\footnotetext{https://www.wolframalpha.com/input?i=y\%3D\%5Cfrac\%7Bk-1\%7D\%7Bk-2\%7Dx-\%5Cfrac\%7B1\%7D\%7Bk\%28k-2\%29\%7D\%2C+1\%2Bx\%2B1\%2Fk\%3D1\%2B1\%2Fz\%2B\%281-y\%29\%2Fw\%2C+2\%2B\%282w-1\%29\%282z\%2B1\%29+\%3D+2k-1\%2C+1\%2Fz\%3D\%281-y\%29\%2Fw, where $y=x_t, x=x_0, z=k', w=k'', k=k$}, we get that for 
$k'=\frac{\sqrt{8k^3-32k+25}+4k-5}{4(k-1)}$, $k''=\frac{\sqrt{8k^3-32k+25}-4k+3}{4(k-2)}$ and $x_0=\frac{\sqrt{8k^3-32k+25}-5k+6}{k(k-1)}$ the equations are satisfied.\footnote{Since $k'$ and $k''$ are integers, we need to consider the rounding up of these values; this does not change the asymptotical runtime guarantee.} Thus, we get a construction time of: 
\begin{align*}
    \Ot(k'm+k''n^{1+\frac{1}{k'}+\frac{1-x_t}{k''}}+kn^{1+x_0+\frac{1}{k}}) 
    &= \Ot(\sqrt{k}m+kn^{1+\frac{\sqrt{8k^3-32k+25}-5k+6}{k(k-1)}+\frac{1}{k}}) \\
    &= \Ot(\sqrt{k}m+kn^{1+\frac{\sqrt{8k^3-32k+25}-4k+5}{k(k-1)}})\\
    &= \Ot(\sqrt{k}m+kn^{1+\frac{2\sqrt{2}}{\sqrt{k}}-\frac{5}{k}+\frac{2\sqrt{2}}{k^{1.5}}+\frac{1}{k^2}-\frac{2\sqrt{2}}{k^{2.5}}+\frac{9}{4\sqrt{2}k^{3.5}}+\frac{1}{k^4}+O(k^{-4.5})}) \\
    &= \Ot(\sqrt{k}m+kn^{1+\frac{2\sqrt{2}}{\sqrt{k}}}),
\end{align*}
Where the last inequality follows from the fact that $-\frac{5}{k}+\frac{2\sqrt{2}}{k^{1.5}}+\frac{1}{k^2}-\frac{2\sqrt{2}}{k^{2.5}}+\frac{9}{4\sqrt{2}k^{3.5}}+\frac{1}{k^4}+O(k^{-4.5}) < 0$ for every $k\ge 1$. Thus, we get a construction time of $O(\sqrt{k}m+kn^{1+\frac{2\sqrt{2}}{\sqrt{k}}})$,
as required.
\Cref{T-construction-spanner-ado} follows from~\Cref{L-Construction-spanner-Almost-Linear},~\Cref{L-Almost-Linear-Space} and~\Cref{L-Almost-Linear-Correctness}.

\ifC \else
\section{Faster constructions for unweighted graphs}\label{S-Construction-unweighted}

In this section, we consider unweighted undirected graphs. In the following two lemmas, we present some preliminaries needed for our new algorithm. Later, in each subsection, we present a new fast construction for distance oracle for unweighted graphs.

We start by improving~\Cref{L-No-Intersect-Weighted} for the case of unweighted graphs. 
\begin{lemma}\label{L-no-intersect-unweighted}
    If $P(u,v)\not\subseteq G_{S}$, then $h(u)+h(v)\le d_G(u,v)+1$
\end{lemma}
\begin{proof}
    Since $P(u,v)\not\subseteq G_{S}$, there is an edge $(x,y)\subseteq P(u,v)$ such that $(x,y)\notin G_{S}$. 
    Since $(x,y)\notin G_{S}$, it follows that $h(x)\le 1$ and $h(y) \le 1$.
    From the triangle inequality, we have that $h(u) \le d(u,x) + h(x) \le d(u,x)+1$, and that $h(v) \le d(y,v)+1$. Summing these inequalities, we get that $h(u)+h(v) \le d(u,x)+1+d(v,y)+1$. 
    From the definition of $P(u,v)$ it follows that $d(u,v) = d(u,x) + 1 + d(y,v)$. Thus, $h(u)+h(v) \le d(u,x)+1+d(v,y)+1 =d(u,v)+1$, as required.
\end{proof}
The following spanner for unweighted graphs replaces the spanner from  
\Cref{L-Spanner-2k-1} in most of our constructions for unweighted graphs.
\begin{lemma}[\cite{DBLP:journals/talg/BaswanaKMP10}]\label{L-Spanner-k-k-1-unweighted}
For any integer $k$, for an  unweighted graph $G$, a $(k, k - 1)$-spanner with $O(n^{1+\frac{1}{k}})$ edges can be computed in $O(m)$ time.
\end{lemma}

\subsection{$(2k-1,2)$-approximation in $km+kn^{\frac{3}{2}+\frac{1}{k-1}-\frac{1}{4k-6}}$ construction time}\label{S-Additive-2}
In this section, we improve the result of \Cref{S-construction-spanner} at the price of allowing an additive error of $2$, we prove:
\begin{theorem}\label{T-Construction-spanner-Unweighted}   
    Let $G=(V, E)$ be an unweighted undirected graph and $k\geq 3$ be an integer. There is an $O(km+kn^{\frac{3}{2}+\frac{1}{k-1}-\frac{1}{4k-6}})$ time algorithm that constructs a $(2k-1,2)$-approximation distance oracle that uses $\Ot(n^{1+\frac{1}{k}})$-space and answers distance queries in $\Ot(k)$ time.
\end{theorem}

First, we describe the construction of the distance oracle. The construction is identical to the construction of \Cref{S-construction-spanner}. Except we use a different value of $k'$ ($k-1$ instead of $k-2$). 
Let $\frac{1}{k}\le x_0<1$ be a parameter to be determined later.
We start by constructing $\hADO(G,k,x_0)$, and save the returned set $S_t$.

Let $k'=k-1$. Using~\Cref{L-Spanner-2k-1} we create a $(2k'-1)$-stretch spanner $H$. 
For every $u\in V$, we add the edge $(u,p_{S_t}(u))$ to $H$ with weight $h_{S_t}(u)$.
For every $s_1,s_2\in S_t$, we compute and store the value of $d_H(s_1,s_2)$. The construction algorithm is presented in \Cref{A-Construction-Weighted-Small-k}.

Given $u,v\in V$, the query algorithm returns $\min(\hADO(G,k,x_0)\Query(u,v), h_{S_t}(u)+d_H(p_{S_t}(u),p_{S_t}(v))+h_{S_t}(v))$. 
Since the construction is identical to \Cref{S-construction-spanner},~\Cref{L-Construction-spanner} and~\Cref{L-ADO-Spanner-Space} proves the space and running time of the distance oracle.
In the following lemma, we prove the approximation of the distance oracle.

\begin{lemma}\label{L-ADO-Spanner-correctness-unweighted}
    $\hat{d}(u,v) \le (2k-1)d(u,v)+2$ 
\end{lemma}
\begin{proof}
    We divide the proof into two cases. 
    Consider the case that $P(u,v) \subseteq G_{S_t}$. From~\Cref{T-Construct-to-T} we have that $\hADO\Query(u,v)\le (2k-1)d(u,v)$, as required.
    
    Consider now the case that $P(u,v) \not\subseteq G_{S_t}$, from~\Cref{L-no-intersect-unweighted} we know that $h(u)+h(v) \le d(u,v)+1$.
    From the triangle inequality we get that $d_H(p(u),p(v)) \le h(u) + d_H(u,v) + h(v)$.
    Since $H$ is a $(2k'-1)$-stretch spanner we know that $d_H(u,v)\le (2k'-1)d(u,v)$, and thus
    $d_H(p(u),p(v)) \le h(u)+(2k'-1)d_G(u,v)+h(v) \le (2k'-1)d(u,v)+d(u,v)+1$. Where the last inequality holds since $h(u)+h(v) \le d(u,v)+1$.
    In the query algorithm, we have that $\hat{d}(u,v) \le h(u)+h(v)+d_H(p(u),p(v)) \le d(u,v)+1+(2k'-1)d(u,v)+d(u,v)+1 \le (2k'+1)d(u,v)+2$. Since  $k'= k-1$, we get that $\hat{d}(u,v) \le (2k-1)d(u,v)+2$, as required.
\end{proof}
Next, we move to choose the value of $x_0$ that minimizes the construction time of our distance oracle.
From~\Cref{L-Construction-spanner} we know that the construction time is $O(k'm+kn^{1+x_0+\frac{1}{k}}+n^{2+\frac{1}{k'}-\frac{k-1}{k-2}x_0+\frac{1}{k(k-2)}})$. To minimize the running time we want to have that $1+x_0+\frac{1}{k} = 2+\frac{1}{k'}-\frac{k-1}{k-2}x_0+\frac{1}{k(k-2)}$. From 
\href{https://www.wolframalpha.com/input?i=y%3D%5Cfrac%7Bk-1%7D%7Bk-2%7Dx-%5Cfrac%7B1%7D%7Bk%28k-2%29%7D%2C+1%2Bx%2B1%2Fk%3D2%2B1%2F%28k-1%29-y}{\footnotemark}
\footnotetext{https://www.wolframalpha.com/input?i=y\%3D\%5Cfrac\%7Bk-1\%7D\%7Bk-2\%7Dx-\%5Cfrac\%7B1\%7D\%7Bk\%28k-2\%29\%7D\%2C+1\%2Bx\%2B1\%2Fk\%3D2\%2B1\%2F\%28k-1\%29-y
where $x=x_0$ and $y=x_t$} we get that for 
$x_0=\frac{k^3-3k^2+4k-3}{k(2k^2-5k+3)}$ the terms are equal. We get a construction time of 
\begin{align*}
     O(k'm+kn^{1+x_0+\frac{1}{k}}+n^{2+\frac{1}{k'}-\frac{k-1}{k-2}x_0+\frac{1}{k(k-2)}}) 
    &= O(km+kn^{1+\frac{k^3-3k^2+4k-3}{k(2k^2-5k+3)}+\frac{1}{k}}) \\
    &= O(km+kn^{1+\frac{k^3-3k^2+4k-3+(2k^2-5k+3)}{k(2k^2-5k+3)}}) \\
    &= O(km+kn^{\frac{3}{2}+\frac{3k-5}{2(k-1)(2k-3)}}) \\
    &= O(km+kn^{\frac{3}{2}+\frac{1}{k-1}-\frac{1}{4k-6}}),
\end{align*}
as required.

\Cref{T-Construction-spanner-Unweighted} follows from~\Cref{L-Construction-spanner},~\Cref{L-ADO-Spanner-Space} and~\Cref{L-ADO-Spanner-correctness-unweighted}.

\subsection{$(2k-1,2k-2)$-approximation in $\approx km+kn^{\frac{3}{2}+\frac{1}{2k}}$ construction time}\label{S-Additive-2k-2}
In this section, we further improve the construction time for small values of $k$ at the cost of increasing the additive error from $2$ to $2k-2$, we prove:
\begin{theorem}\label{T-Construction-spanner-Unweighted-Additive}
    Let $G=(V, E)$ be an unweighted undirected graph and $k\geq 3$ be an integer.  There is an $O((km+kn^{\frac{3}{2}+\frac{1}{2k-3}-\frac{1}{2(2k-3)^2}})\log\log{n})$ time algorithm that constructs a $(2k-1,2k-2)$-approximation distance oracle that uses $\Ot(n^{1+\frac{1}{k}})$-space and answers distance queries in $\Ot(k)$ time.
\end{theorem}

The construction is the same as the one of~\Cref{S-Additive-2}, but instead of using the $(2k-1)$-stretch spanner of~\Cref{L-Spanner-2k-1} we use the $(k,k-1)$-approximation spanner of~\Cref{L-Spanner-k-k-1-unweighted}. This substitution allows us to have $k'=2k-3$ (instead of $k'-1$ as in the previous section),  
Let $\frac{1}{k}\le x_0<1$ be a parameter to be determined later.
We start by constructing $\hADO(G,k,x_0)$, and save the returned set $S_t$.

Let $k'=2k-3$. Using~\Cref{L-Spanner-k-k-1-unweighted} we create a $(k',k'-1)$-stretch spanner $H$. 
For every $u\in V$, we add the edge $(u,p_{S_t}(u))$ to $H$ with weight $h_{S_t}(u)$.
For every $s_1,s_2\in S_t$ we compute and store the value of $d_H(s_1,s_2)$. The construction algorithm is presented in \Cref{A-Construction-Weighted-Small-k}.

Given $u,v\in V$, the query algorithm returns $\min(\hADO(G,k,x_0)\Query(u,v), h_{S_t}(u)+d_H(p_{S_t}(u),p_{S_t}(v))+h_{S_t}(v))$. 

Since the construction is identical to \Cref{S-construction-spanner} (and the spanner of~\Cref{L-Spanner-k-k-1-unweighted} is computed in $O(m)$ time),~\Cref{L-Construction-spanner} and~\Cref{L-ADO-Spanner-Space} proves the space and running time of the distance oracle.
Next, in the following lemma we prove the approximation of the distance oracle.
\begin{lemma}\label{L-Correctness-spanner-Unweighted-Additive}
    $\hat{d}(u,v) \le (2k-1)d(u,v)+2k-2$
\end{lemma}
\begin{proof}
    We divide the proof into two cases. 
    Consider the case that $P(u,v) \subseteq G_{S_t}$. From~\Cref{T-Construct-to-T} we have that $\hADO\Query(u,v)\le (2k-1)d(u,v)$, as required.
    
    Consider now the case that $P(u,v) \not\subseteq G_{S_t}$, from~\Cref{L-no-intersect-unweighted} we know that $h(u)+h(v) \le d(u,v)+1$.
    From the triangle inequality we get that $d_H(p(u),p(v)) \le h(u) + d_H(u,v) + h(v)$.
    Since $H$ is a $(k',k'-1)$-approximation spanner we know that $d_H(u,v)\le k'd(u,v)+k'-1$, and thus
    $d_H(p(u),p(v)) \le h(u)+k'd(u,v)+k'-1+h(v) \le k'd(u,v)+k'-1 + d(u,v)+1$. Where the last inequality holds since $h(u)+h(v) \le d(u,v)+1$.
    
    In the query algorithm, we have that $\hat{d}(u,v) \le h(u)+h(v)+d_H(p(u),p(v)) \le d(u,v)+1+k'd(u,v)+k'-1 + d(u,v)+1 \le (k'+2)d(u,v) +k'+1$. Since  $k'= 2k-3$, we get that $\hat{d}(u,v) \le (2k-1)d(u,v)+2k-2$, as required.
\end{proof}
Next, we move to choose the value of $x_0$ that minimizes the construction time of our distance oracle.
From~\Cref{L-Construction-spanner} we know that the construction time is $O(k'm+kn^{1+x_0+\frac{1}{k}}+n^{2+\frac{1}{k'}-\frac{k-1}{k-2}x_0+\frac{1}{k(k-2)}})$. To minimize the running time we want to have that $1+x_0+\frac{1}{k} = 2+\frac{1}{k'}-\frac{k-1}{k-2}x_0+\frac{1}{k(k-2)}$. From \href{https://www.wolframalpha.com/input?i=y%3D%5Cfrac%7Bk-1%7D%7Bk-2%7Dx-%5Cfrac%7B1%7D%7Bk%28k-2%29%7D%2C+1%2Bx%2B1%2Fk%3D2%2B1%2F%282k-3%29-y}{\footnotemark}
\footnotetext{https://www.wolframalpha.com/input?i=y\%3D\%5Cfrac\%7Bk-1\%7D\%7Bk-2\%7Dx-\%5Cfrac\%7B1\%7D\%7Bk\%28k-2\%29\%7D\%2C+1\%2Bx\%2B1\%2Fk\%3D2\%2B1\%2F\%282k-3\%29-y,
where $x=x_0$ and $y=x_t$}
we get that for $x_0=\frac{2k^3-8k^2+13k-9}{k(2k-3)^2}$ the terms are equal. We get a construction time of 
\begin{align*}
     O(k'm+kn^{1+x_0+\frac{1}{k}}+n^{2+\frac{1}{k'}-\frac{k-1}{k-2}x_0+\frac{1}{k(k-2)}}) 
    &= O(km+kn^{1+\frac{2k^3-8k^2+13k-9}{k(2k-3)^2}+\frac{1}{k}}) \\
    &= O(km+kn^{1+\frac{2k^3-8k^2+13k-9+(4k^2-6k+9)}{k(2k-3)^2}})\\
    &= O(km+kn^{1+\frac{2k^2-4k+1}{(2k-3)^2}})\\
    &= O(km+kn^{\frac{3}{2}+\frac{1}{2k-3}-\frac{1}{2(2k-3)^2}}),
\end{align*}
as required.
\Cref{T-Construction-spanner-Unweighted-Additive} follows from~\Cref{L-Construction-spanner},~\Cref{L-ADO-Spanner-Space} and~\Cref{L-Correctness-spanner-Unweighted-Additive}.

\subsection{$(2k-1, 2k-1)$-approximation in $m+n^{1+\frac{2}{\sqrt{k}}}$ - construction time}\label{S-unweighted-large-k}
In this section, we show that the construction time of \Cref{S-construction-spanner-ado} can be improved in the unweighted case by adding an additive error. We prove:
\begin{theorem} \label{T-construction-spanner-ado-unweighted}
    Let $G=(V, E)$ be an unweighted undirected graph and let $k\geq 16$ be an integer. There is an $O(\sqrt{k}m+kn^{1+\frac{k\sqrt{16k^3-15k^2-32k+32}-3k^2+4k}{k(2k^2-k-2)}})\leq O(\sqrt{k}m+n^{1+\frac{2}{\sqrt{k}}})$ time algorithm that constructs a $(2k-1,2k-1)$-approximation distance oracle that uses $\Ot(n^{1+\frac{1}{k}})$-space and answers distance queries in $\Ot(k)$ time.
\end{theorem}
The construction is the same as in \Cref{S-construction-spanner-ado}, but instead of using the spanner of~\Cref{L-Spanner-2k-1}, we use the spanner of~\Cref{L-Spanner-k-k-1-unweighted}, which results in a different constraint on $k'$ and $k''$ and thus an improved result at the cost of an additive error.

let $\frac{1}{k}\le x_0<1$ be a parameter to be determined later.
We start by constructing $\hADO(G,k,x_0)$, and save the returned set $S_t$.
Let $k'$ and $k''$ be two integer parameters to be determined later. Using~\Cref{L-Spanner-k-k-1-unweighted} we create a $(k',k'-1)$-stretch spanner $H$, which we  augment with  an edge $(u,p_{S_t}(u))$ of weight $h_{S_t}(u)$, for every $u\in V$.
We construct $\ADO_{P'}(H,k'', S_t)$.
The construction algorithm is presented in \Cref{A-Construction-Weighted-Large-k}.

Given $u,v\in V$, the query algorithm returns: $$\min(\hADO(G,k,x_0)\Query(u,v), h_{S_t}(u)+\ADO_{P'}(H,k'',S_t)\Query(p_{S_t}(u),p_{S_t}(v))+h_{S_t}(v))$$

Since the construction is identical to \Cref{S-construction-spanner-ado} (and the spanner of~\Cref{L-Spanner-k-k-1-unweighted} is computed in $O(m)$ time),~\Cref{L-Construction-spanner-Almost-Linear} and~\Cref{L-Almost-Linear-Space} proves the space and running time of the distance oracle.
The following lemma proves the approximation ratio of the distance oracle.

\begin{lemma}\label{L-Almost-Linear-Correctness-unweighted}
    If $1 + (2k''-1)(k'+1) \le 2k-1$ then $\hat{d}(u,v) \le (2k-1)d(u,v)+2k-1$
\end{lemma}
\begin{proof}
    From the query algorithm we know that $\hat{d}(u,v) \le h(u)+\ADO_{H}(p(u),p(v))+h(v)$.
    From~\Cref{L-no-intersect-unweighted} it follows that $h(u)+h(v) \le d(u,v)+1$.
    From the triangle inequality, and the fact that $H$ is a $(k',k'-1)$-approximation spanner it follows that
    $d_H(p(u),p(v)) \le h(u)+(k')d(u,v)+k'-1+h(v)\le (k'+1)d(u,v)+k'$. Since $\ADO_{P'}$ is a $(2k''-1)$-approximation distance oracle, it follows that
    $\ADO_{P'}(p(u),p(v)) \le (2k''-1)d_H(p(u),p(v)) \le (2k''-1)((k'+1)d(u,v)+k')$.
    Thus, 
    $\hat{d}(u,v) \le h(u)+h(v) + (2k''-1)((k'+1)d(u,v)+k') \le d(u,v)+1+(2k''-1)((k'+1)d(u,v)+k')$.
    Thus if $1 + (2k''-1)(k'+1) \le 2k-1$ then $\hat{d}(u,v) \le (2k-1)d(u,v)+2k-1$, as required.
\end{proof}
Next, we move to choose the value of the parameters $x_0,k', k''$ that minimizes the running time of our algorithm while maintaining the correctness of the query.

From~\Cref{L-Construction-spanner-Almost-Linear} we have that the construction time is $O(k'm+k''n^{1+\frac{1}{k'}+\frac{1-\frac{k-1}{k-2}x_0+\frac{1}{k(k-2)}}{k''}}+kn^{1+x_0+\frac{1}{k}})$. In addition, for the correctness,~\Cref{L-Almost-Linear-Correctness-unweighted} requires that 
$1 + (2k''-1)(k'+1) \le 2k-1$.
To minimize the term $n^{1+\frac{1}{k'}+\frac{1-\frac{k-1}{k-2}x_0+\frac{1}{k(k-2)}}{k''}}$ regarding $k'$ and $k''$ we want to have that $\frac{1}{k'}=\frac{1-\frac{k-1}{k-2}x_0+\frac{1}{k(k-2)}}{k''}$.
Thus, to find the optimal values for $x_0,k',k''$ we need to solve the following equations:
\begin{itemize}
    \item $1+\frac{1}{k'}+\frac{1-\frac{k-1}{k-2}x_0+\frac{1}{k(k-2)}}{k''} = 1+x_0+\frac{1}{k}$
    \item $1 + (2k''-1)(k'+1) = 2k-1$
    \item $\frac{1}{k'}=\frac{1-\frac{k-1}{k-2}x_0+\frac{1}{k(k-2)}}{k''}$
\end{itemize}
From \href{https://www.wolframalpha.com/input?i=y\%3D\%5Cfrac\%7Bk-1\%7D\%7Bk-2\%7Dx-\%5Cfrac\%7B1\%7D\%7Bk\%28k-2\%29\%7D\%2C+1\%2Bx\%2B1\%2Fk\%3D1\%2B1\%2Fz\%2B\%281-y\%29\%2Fw\%2C+1\%2B\%282w-1\%29\%28z\%2B1\%29+\%3D+2k-1\%2C+1\%2Fz\%3D\%281-y\%29\%2Fw}{\footnotemark}
\footnotetext{https://www.wolframalpha.com/input?i=y\%3D\%5Cfrac\%7Bk-1\%7D\%7Bk-2\%7Dx-\%5Cfrac\%7B1\%7D\%7Bk\%28k-2\%29\%7D\%2C+1\%2Bx\%2B1\%2Fk\%3D1\%2B1\%2Fz\%2B\%281-y\%29\%2Fw\%2C+1\%2B\%282w-1\%29\%28z\%2B1\%29+\%3D+2k-1\%2C+1\%2Fz\%3D\%281-y\%29\%2Fw, where $y=x_t, x=x_0, z=k', w=k'', k=k$}, we get that for 
$k''=\frac{\sqrt{16k^3-15k^2-32k+32}-5k+4}{4(k-2)}$,
$k'=\frac{\sqrt{16k^3-15k^2-32k+32}+3k-4}{4(k-1)}$ and 
$x_0=\frac{k\sqrt{16k^3-15k^2-32k+32}-5k^2+5k+2}{k(2k^2-k-2)}$ the equations are satisfied.\footnote{Since $k'$ and $k''$ are integers, we need to consider the rounding up of these values; this does not change the asymptotical runtime guarantee.} Thus, we get a construction time of: 

\begin{align*}
    O(k'm+k''n^{1+\frac{1}{k'}+\frac{1-x_t}{k''}}+kn^{1+x_0+\frac{1}{k}}) 
    &= O(\sqrt{k}m+kn^{1+\frac{k\sqrt{16k^3-15k^2-32k+32}-5k^2+5k+2}{k(2k^2-k-2)}+\frac{1}{k}}) \\
    &= O(\sqrt{k}m+kn^{1+\frac{k\sqrt{16k^3-15k^2-32k+32}-3k^2+4k}{k(2k^2-k-2)}}) \\
    &= O(\sqrt{k}m+kn^{1+\frac{\sqrt{16k^3-15k^2-32k+32}-3k+4}{2k^2-k-2}}) \\
    &\stackrel{\footnotemark}{=}O(\sqrt{k}m+kn^{1+\frac{2}{\sqrt{k}}-\frac{3}{2k}+\frac{1}{16k\sqrt{k}}+\frac{5}{k^2}-\frac{193}{1024k^{2.5}} - \frac{7}{8k^3}+O(k^{-3.5})}) \\
    &= O(\sqrt{k}m+kn^{1+\frac{2}{\sqrt{k}}})
\end{align*}
\footnotetext{\href{https://www.wolframalpha.com/input?i=\%5Cfrac\%7B\%5Csqrt\%7B16k\%5E3-15k\%5E2-32k\%2B32\%7D-3k\%2B4\%7D\%7B\%282k\%5E2-k-2\%29\%7D+limit+for+k\%3Dinfty}{https://www.wolframalpha.com/input?i=\%5Cfrac\%7B\%5Csqrt\%7B16k\%5E3-15k\%5E2-32k\%2B32\%7D-3k\%2B4\%7D\%7B\%282k\%5E2-k-2\%29\%7D+limit+for+k\%3Dinfty}}
Where the last inequality follows from the fact that $-\frac{3}{2k}+\frac{1}{16k\sqrt{k}}+\frac{5}{k^2}-\frac{193}{1024k^{2.5}} - \frac{7}{8k^3}+O(k^{-3.5}) < 0$ for every $k\ge 4$. Thus, we get a construction time of $O(\sqrt{k}m+kn^{1+\frac{2}{\sqrt{k}}})$,
as required.~\Cref{T-construction-spanner-ado-unweighted} follows from~\Cref{L-Construction-spanner-Almost-Linear},~\Cref{L-Almost-Linear-Space} and~\Cref{L-Almost-Linear-Correctness-unweighted}.

\subsection{$((2k-1)(1+\eps),(2k-1)\beta)$-approximation in $mn^\eps+n^{1+\frac{1}{k}+\eps}$ construction time}
In this section, we present a construction of a simple $((2k-1)(1+\eps),(2k-1)\beta)$ distance oracle using $O(kn^{1+\frac{1}{k}})$ space.
We prove:
\begin{theorem} \label{T-Unweighted-2k-1+eps}
    There is an $O(kn^{1+\frac{1}{k}})$-space $((2k-1)(1+\eps),(2k-1)\beta)$-approximation distance oracle with $O(k)$ query time. The distance oracle is constructed in $O(mn^\eps+n^{1+\eps+1/k})$ time. 
\end{theorem}
\begin{proof}
    First, we construct $H$, a $(1+\eps, \beta)$-approximation spanner of~\cite{DBLP:conf/soda/ThorupZ06} in $O(mn^{\eps})$ time.
    Next, we construct the $\ADO(H,k)$, which uses $O(n^{1+\frac{1}{k}})$ space, in $O(|E(H)|\cdot n^{\frac{1}{k}})=O(n^{1+\frac{1}{k}+\eps})$ time. Overall we get a construction time of $O(mn^\eps+n^{1+1/k+\eps})$ and $O(n^{1+\frac{1}{k}})$ space, as required.
    The query algorithm returns $\ADO(H,k)\Query(u,v)$, and since $H$ is a $(1+\eps,\beta)$-approximation spanner we get that $\ADO(H,k)\Query(u,v) \leq (2k-1)d_H(u,v) \leq (2k-1)(1+\eps)d(u,v)+(2k-1)\beta$, as required.
\end{proof}

\fi
\bibliography{articles}

\clearpage
\appendix

\section{Comparison with previous results}\label{S-comparison}
\subsection{Comparison for weighted graphs }
\begin{table}[H]
    \centering
    \begin{tabular}{|c|c|c|}
    \hline
         Construction time & Reference & Comments \\
         \hline\hline
         $O(kmn^{1/k})$ &  \cite{DBLP:journals/jacm/ThorupZ05} 
 &  \\
 \hline
        $O(\max(n^{1+2/k}, m^{1-\frac{1}{k-1}}n^{\frac{2}{k-1}})\log\log{n})$ & Theorem~\ref{T-Construction-no-spanner} & $n^{1+1/k}<m<n^2$ \\
        \Xhline{3\arrayrulewidth}
        $\Ot(n^2)$  &~\cite{DBLP:journals/siamcomp/BaswanaK10} & \\
        \hline
        $\Ot(km+kn^{\frac{3}{2}+\frac{2}{k}+O(k^{-2})})$ &~\cite{DBLP:conf/soda/Wulff-Nilsen12} & $k\ge 6$ \\
        \hline

        $\Ot(km+kn^{\frac{3}{2}+\frac{3}{4k-6}})$ & Theorem~\ref{T-Construction-spanner} & $k\ge 4$\\
        \Xhline{3\arrayrulewidth}
        $\Ot(\sqrt{k}m+kn^{1+\frac{2\sqrt{6}}{\sqrt{k}}+O(k^{-1})})$ &~\cite{DBLP:conf/soda/Wulff-Nilsen12} & $k\ge 96$ \\ 
        \hline
        $\Ot(\sqrt{k}m+kn^{1+\frac{2\sqrt{2}}{\sqrt{k}}})$ & Theorem~\ref{T-construction-spanner-ado} & $k\ge 16$ \\ 
        \hline
    \end{tabular}
    \caption{$(2k-1)$-stretch distance oracle constructions for weighted graphs with $O(kn^{1+1/k})$ space.}
    \label{tab:construction}
\end{table}

\begin{figure}[H]
\begin{tikzpicture}
    \definecolor{color0}{rgb}
    {0.15,0.7,0.15}
    \definecolor{color1}{rgb}
    {0,0,1}
   
    \begin{axis}[
    width=0.8\textwidth,
    height=0.4\textwidth, 
    legend cell align={left},
    legend style={fill opacity=0.8, draw opacity=1, text opacity=1, draw=white!80!black,
    row sep=5pt},
    tick align=outside,
    tick pos=left,
    x grid style={white!86.2745098039216!black},
    xlabel={$k$},
    xmajorgrids,
    xmin=3, xmax=15,
    xtick style={color=black},
    xtick={1,2,3,4,5,6,7,8,9,10,11,12,13,14,15,16,17,18,19,20,21,22,23,24,25,26,27,28,29,30},
    y grid style={white!86.2745098039216!black},
    ylabel style={rotate=-90},
    ylabel={$f(k)$},
    ymajorgrids,
    ymin=0.5, ymax=1,
    ytick style={color=black},
    ytick={0.5,0.6,0.7,0.8,0.9,1.0}
    ]
    \legend{
        \textcolor{color0}{$\bullet$}~\cite{DBLP:conf/soda/Wulff-Nilsen12},
        \textcolor{color1}{$\bullet$}~Theorem~\ref{T-Construction-spanner},
    }
    
    
        \addplot[
            thick,
            color=color0, 
            dotted
        ]
        coordinates {
            (6,{(1/2 +  2/6)})
            (7,{(1/2 +  1/(2*7) + 3/(2*5))})
            (8,{(1/2 +  2/6 - 1/(8*6))})
            (9,{(1/2 +  2/9)})
            (10,{(1/2 +  1/(2*10) + 3/(2*8))})
            (11,{(1/2 +  2/9 - 1/(11*9))})
            (12,{(1/2 +  2/12)})
            (13,{(1/2 +  1/(2*13) + 3/(2*11))})
            (14,{(1/2 +  2/12 - 1/(14*12))})
            (15,{(1/2 +  2/15)})
            (16,{(1/2 +  1/(2*16) + 3/(2*14))})
            (17,{(1/2 +  2/15 - 1/(17*15))})
        };
        
    
    \addplot [thick, 
            domain=4:17,
            samples=100,
            color1,
            dotted]{min(2,1/2+3/(4*x-6))};
    \addplot[
        only marks,
        mark=*,
        color=color1,
        mark options={scale=1.2}
    ]
    coordinates {
        (4, {min(2, 1/2 +  3/(4*4 - 6))})
        (5, {min(2, 1/2 +  3/(4*5 - 6))})
        (6, {min(2, 1/2 +  3/(4*6 - 6))})
        (7, {min(2, 1/2 +  3/(4*7 - 6))})
        (8, {min(2, 1/2 +  3/(4*8 - 6))})
        (9, {min(2, 1/2 +  3/(4*9 - 6))})
        (10, {min(2, 1/2 +  3/(4*10 - 6))})
        (11, {min(2, 1/2 +  3/(4*11 - 6))})
        (12, {min(2, 1/2 +  3/(4*12 - 6))})
        (13, {min(2, 1/2 +  3/(4*13 - 6))})
        (14, {min(2, 1/2 +  3/(4*14 - 6))})
        (15, {min(2, 1/2 +  3/(4*15 - 6))})
        (16, {min(2, 1/2 +  3/(4*16 - 6))})
        (17, {min(2, 1/2 +  3/(4*17 - 6))})
    };
    \addplot[
            only marks,
            mark=*, 
            color=color0, 
            mark options={scale=1.2}, 
            dotted
        ]
        coordinates {
            (6,{(1/2 +  2/6)})
            (7,{(1/2 +  1/(2*7) + 3/(2*5))})
            (8,{(1/2 +  2/6 - 1/(8*6))})
            (9,{(1/2 +  2/9)})
            (10,{(1/2 +  1/(2*10) + 3/(2*8))})
            (11,{(1/2 +  2/9 - 1/(11*9))})
            (12,{(1/2 +  2/12)})
            (13,{(1/2 +  1/(2*13) + 3/(2*11))})
            (14,{(1/2 +  2/12 - 1/(14*12))})
            (15,{(1/2 +  2/15)})
            (16,{(1/2 +  1/(2*16) + 3/(2*14))})
            (17,{(1/2 +  2/15 - 1/(17*15))})
        };

    \end{axis}
   
    \end{tikzpicture}

  \caption{Comparison between Theorem~\ref{T-Construction-spanner}: $f(k)=1/2+\frac{3}{4k-6}$ and~\cite{DBLP:conf/soda/Wulff-Nilsen12}: $f(k)=1/2+\frac{2}{k}+O(1/k^2)$}
  
  \label{fig:theorem_weighted}
\end{figure}
\newcommand{\TheoremTwof}[1]{n^{\pgfmathparse{3/2+3/(4*#1-6)}\pgfmathprintnumberto[precision=3]{\pgfmathresult}{\roundednumber}\roundednumber}}

\begin{table}[H]
    \centering
    \begin{tabular}{|c|c||c|c|c|c|c|c|c|c|c|}
         \hline
          & $k$ & $2$ & $3$ & $4$ & $5$ & $6$ & $7$ & $8$ & $9$ & $10$ \\
            \hline\hline
         \cite{DBLP:journals/siamcomp/BaswanaK10,DBLP:conf/soda/Wulff-Nilsen12} & \multirow{ 2}{*}{$n^{1+f(k)}$}& $n^2$ & $n^2$ & $n^2$ & $n^2$ & $n^{1.833}$ & $n^{1.821}$ & $n^{1.8125}$ & $n^{1.722}$ & $n^{1.717}$ \\
         Theorem~\ref{T-Construction-spanner} & &  $-$ & $-$ & $\TheoremTwof{4}$ & $\TheoremTwof{5}$ & $\TheoremTwof{6}$ & $\TheoremTwof{7}$ & $\TheoremTwof{8}$ & $\TheoremTwof{9}$ & $\TheoremTwof{10}$ \\ 
         \hline
    \end{tabular}
    \caption{The sparsest graphs for which a linear time construction algorithm exists for $(2k-1)$-stretch distance oracles with $O(n^{1+1/k})$-space in weighted graphs.}
    \label{tab:examples}
\end{table}

          

\clearpage
\ifC \else
\subsection{Comparison for unweighted graphs}
\begin{table}[H]
    \centering
    \begin{tabular}{|c|c|c|c|}
         \hline
         Construction time &  Approximation & Reference & Comments \\
         \hline\hline
        $\Ot(km+kn^{\frac{3}{2}+\frac{1}{k} + \frac{1}{k(2k-2)}})$ & $(2k-1)\delta+2$ & \cite{DBLP:conf/icalp/BaswanaGSU08} & \\ \hline
        $\Ot(km+kn^{\frac32+\frac{1}{k-1}-\frac{1}{4k-6}})$ & $(2k-1)\delta+2$ & Theorem~\ref{T-Construction-spanner-Unweighted} & $k\ge 3$\\ \hline
        $\Ot(km+kn^{\frac{3}{2}+\frac{1}{2k}+\frac{3.5k-4.5}{k(4k^2-8k+3)}})$ & $(2k-1)\delta+2k-2$ & Theorem~\ref{T-Construction-spanner-Unweighted-Additive} & $k\ge 3$ \\
        \Xhline{3\arrayrulewidth}
        $\Ot(\sqrt{k}m+kn^{1+\frac{2\sqrt{6}}{\sqrt{k}}+O(k^{-1})})$ & $(2k-1)\delta$ &~\cite{DBLP:conf/soda/Wulff-Nilsen12} & $k\ge 96$ \\ \hline
        $\Ot(\sqrt{k}m+kn^{1+\frac{2\sqrt{2}}{\sqrt{k}}})$ & $(2k-1)\delta$ & Theorem~\ref{T-construction-spanner-ado} & $k\ge 16$ \\ \hline
        $\Ot(\sqrt{k}m+kn^{1+\frac{2}{\sqrt{k}}})$ & $(2k-1)\delta+2k-1$ & Theorem~\ref{T-construction-spanner-ado-unweighted} & $k\ge 13$ \\
        \Xhline{3\arrayrulewidth}
        $\Ot(km+n^{1+1/k+\eps})$ & $O(k)\cdot \delta$ &~\cite{DBLP:conf/soda/Wulff-Nilsen12} & \\ \hline
        $\Ot(\frac{1}{\eps}mn^{\eps}+kn^{1+1/k+\eps})$ & $(2k-1)(1+\eps)\delta+(2k-1)\beta$ & Theorem~\ref{T-Unweighted-2k-1+eps} & $\beta=f(\eps)$\\
        \hline
    \end{tabular}
    \caption{Distance oracle constructions for dense unweighted graphs with $O(kn^{1+1/k})$ space. $\delta = d(u,v)$.}
    \label{tab:construction_unweighted}
\end{table}

\begin{figure}[H]
    \begin{tikzpicture}
    \definecolor{color0}{rgb}
    {0.15,0.7,0.15}
    \definecolor{color1}{rgb}
    {0.529, 0.808, 0.922}
    \definecolor{color2}{rgb}{0,0,1}
   
    \begin{axis}[
    width=0.8\textwidth,
    height=0.4\textwidth, 
    legend cell align={left},
    legend style={fill opacity=0.8, draw opacity=1, text opacity=1, draw=white!80!black,
    row sep=5pt},
    tick align=outside,
    tick pos=left,
    x grid style={white!86.2745098039216!black},
    xlabel={$k$},
    xmajorgrids,
    xmin=2, xmax=15,
    xtick style={color=black},
    xtick={1,2,3,4,5,6,7,8,9,10,11,12,13,14,15,16},
    y grid style={white!86.2745098039216!black},
    ylabel style={rotate=-90},
    ylabel={$f(k)$},
    ymajorgrids,
    ymin=0.5, ymax=1,
    ytick style={color=black},
    ytick={0.5,0.6,0.7,0.8,0.9}
    ]
    \addplot [thick, 
            domain=3:17,
            samples=100,
            color0,
            dotted]{1/2+1/x+1/(x*(2*x-2))};
        
    \legend{
        \textcolor{color0}{$\bullet$}~\cite{DBLP:conf/icalp/BaswanaGSU08},
        \textcolor{color1}{$\bullet$}~Theorem~\ref{T-Construction-spanner-Unweighted},
        \textcolor{color2}{$\bullet$}~Theorem~\ref{T-Construction-spanner-Unweighted-Additive}
    }
    \addplot [thick, 
            domain=3:17,
            samples=100,
            color1,
            dotted]{1/2+1/(x-1)-(1)/(4*x-6)};

        \addplot [thick, 
            domain=3:17,
            samples=100,
            color2,
            dotted]{1/2+1/(2*x-3)-1/(2*(2*x-3)^2)};

    \addplot[
        only marks,
        mark=*,
        color=color1,
        mark options={scale=1}
    ]
    coordinates {
            (3 ,  {1/2+1/(3-1)-(1)/(4*3-6)})
            (4 ,  {1/2+1/(4-1)-(1)/(4*4-6)})
            (5 ,  {1/2+1/(5-1)-(1)/(4*5-6)})
            (6 ,  {1/2+1/(6-1)-(1)/(4*6-6)})
            (7 ,  {1/2+1/(7-1)-(1)/(4*7-6)})
            (8 ,  {1/2+1/(8-1)-(1)/(4*8-6)})
            (9 ,  {1/2+1/(9-1)-(1)/(4*9-6)})
            (10 ,  {1/2+1/(10-1)-(1)/(4*10-6)})
            (11 ,  {1/2+1/(11-1)-(1)/(4*11-6)})
            (12 ,  {1/2+1/(12-1)-(1)/(4*12-6)})
            (13 ,  {1/2+1/(13-1)-(1)/(4*13-6)})
            (14 ,  {1/2+1/(14-1)-(1)/(4*14-6)})
            (15 ,  {1/2+1/(15-1)-(1)/(4*15-6)})
            (16 ,  {1/2+1/(16-1)-(1)/(4*16-6)})
            (17 ,  {1/2+1/(17-1)-(1)/(4*17-6)})
            
        };
    \addplot[
            only marks,
            mark=*, 
            color=color0, 
            mark options={scale=1} 
        ]
        coordinates {
            (3 ,  {1/2 + 1/3 + 1/(3*(2*3 - 2))})
            (4 ,  {1/2 + 1/4 + 1/(4*(2*4 - 2))})
            (5 ,  {1/2 + 1/5 + 1/(5*(2*5 - 2))})
            (6 ,  {1/2 + 1/6 + 1/(6*(2*6 - 2))})
            (7 ,  {1/2 + 1/7 + 1/(7*(2*7 - 2))})
            (8 ,  {1/2 + 1/8 + 1/(8*(2*8 - 2))})
            (9 ,  {1/2 + 1/9 + 1/(9*(2*9 - 2))})
            (10,  {1/2 + 1/10 + 1/(10*(2*10 - 2))})
            (11,  {1/2 + 1/11 + 1/(11*(2*11 - 2))})
            (12,  {1/2 + 1/12 + 1/(12*(2*12 - 2))})
            (13,  {1/2 + 1/13 + 1/(13*(2*13 - 2))})
            (14,  {1/2 + 1/14 + 1/(14*(2*14 - 2))})
            (15,  {1/2 + 1/15 + 1/(15*(2*15 - 2))})
            (16,  {1/2 + 1/16 + 1/(16*(2*16 - 2))})
            (17,  {1/2 + 1/17 + 1/(17*(2*17 - 2))})
        };
    \addplot[
        only marks,
        mark=*, 
        color=color2, 
        mark options={scale=1} 
    ]
    coordinates {
(3 ,  {1/2+1/(2*3-3)-1/(2*(2*3-3)^2)})
(4 ,  {1/2+1/(2*4-3)-1/(2*(2*4-3)^2)})
(5 ,  {1/2+1/(2*5-3)-1/(2*(2*5-3)^2)})
(6 ,  {1/2+1/(2*6-3)-1/(2*(2*6-3)^2)})
(7 ,  {1/2+1/(2*7-3)-1/(2*(2*7-3)^2)})
(8 ,  {1/2+1/(2*8-3)-1/(2*(2*8-3)^2)})
(9 ,  {1/2+1/(2*9-3)-1/(2*(2*9-3)^2)})
(10 ,  {1/2+1/(2*10-3)-1/(2*(2*10-3)^2)})
(11 ,  {1/2+1/(2*11-3)-1/(2*(2*11-3)^2)})
(12 ,  {1/2+1/(2*12-3)-1/(2*(2*12-3)^2)})
(13 ,  {1/2+1/(2*13-3)-1/(2*(2*13-3)^2)})
(14 ,  {1/2+1/(2*14-3)-1/(2*(2*14-3)^2)})
(15 ,  {1/2+1/(2*15-3)-1/(2*(2*15-3)^2)})
(16 ,  {1/2+1/(2*16-3)-1/(2*(2*16-3)^2)})
(17 ,  {1/2+1/(2*17-3)-1/(2*(2*17-3)^2)})
        
    };

    \end{axis}
   
    \end{tikzpicture}

  \caption{Comparison between Theorem~\ref{T-Construction-spanner-Unweighted}: $f(k)=\frac12+\frac{3}{4k}+O(\frac{1}{k^2})$, Theorem~\ref{T-Construction-spanner-Unweighted-Additive}: $f(k)=\frac{1}{2}+\frac{1}{2k}+O(\frac{1}{k^2})$ and~\cite{DBLP:conf/icalp/BaswanaGSU08}: $f(k)=\frac{1}{2}+\frac{1}{k}+O(\frac{1}{k^2})$}
  
  \label{fig:theorem_unweighted}
\end{figure}
\newcommand{\TheoremThreeOf}[1]{n^{\pgfmathparse{3/2+1/(#1-1)-(1)/(4*#1-6))}\pgfmathprintnumberto[precision=3]{\pgfmathresult}{\roundednumber}\roundednumber}}

\newcommand{\PrevTheoremBGSU}[1]{n^{\pgfmathparse{3/2+1/#1+1/(#1*(2*#1-2))}\pgfmathprintnumberto[precision=3]{\pgfmathresult}{\roundednumber}\roundednumber}}
\begin{table}[H]
    \centering
    \begin{tabular}{|c|c||c|c|c|c|c|c|c|c|c|}
         \hline
          & $k$  & $3$ & $4$ & $5$ & $6$ & $7$ & $8$ & $9$ & $10$ \\
            \hline    \hline
         \cite{DBLP:conf/icalp/BaswanaGSU08} & \multirow{ 2}{*}{$n^{1+f(k)}$} & $\PrevTheoremBGSU{3}$ & $\PrevTheoremBGSU{4}$ & $\PrevTheoremBGSU{5}$ & $\PrevTheoremBGSU{6}$ & $\PrevTheoremBGSU{7}$ & $\PrevTheoremBGSU{8}$ & $\PrevTheoremBGSU{9}$ & $\PrevTheoremBGSU{10}$ \\
         Theorem~\ref{T-Construction-spanner-Unweighted}& & $\TheoremThreeOf{3}$ & $\TheoremThreeOf{4}$ & $\TheoremThreeOf{5}$ & $\TheoremThreeOf{6}$ & $\TheoremThreeOf{7}$ & $\TheoremThreeOf{8}$ & $\TheoremThreeOf{9}$ & $\TheoremThreeOf{10}$ \\ 
         \hline
    \end{tabular}
    \caption{The sparsest graphs for which a linear time construction algorithm exists for $(2k-1,2)$-stretch distance oracles with $O(n^{1+1/k})$-space in unweighted graphs.}
    \label{tab:examples-unweighted}
\end{table}

\fi

\end{document}